\documentclass{llncs}
\usepackage{amsmath}
\usepackage{amsfonts}
\usepackage{amssymb}

\usepackage[pdftex]{hyperref}
\usepackage{tikz}
\usepackage[T1]{fontenc}
\usepackage{enumerate}
\usepackage{complexity}
\usepackage{mathtools}
\usepackage{relsize}

\usepackage{stmaryrd} 
\usepackage{enumitem}

\spnewtheorem{examp}{Example}{\bfseries}{\normalfont}

\newcommand{\NN}{\mathbb{N}}

\newcommand{\mhml}{\text{\sc mHML}}
\newcommand{\shml}{\text{\sc sHML}}
\newcommand{\muhml}{\mu\text{HML}}
\newcommand{\chml}{\text{\sc cHML}}

\newcommand{\defeq}{\stackrel{\mathclap{def}}{=}}

\newcommand{\mycap}[1]{\text{\sc#1}}
\newcommand{\acc}{\text{\bf acc}}
\newcommand{\rej}{\text{\bf rej}}

\newcommand{\msf}[1]{\llparenthesis#1\rrparenthesis}

\newcommand{\rec}{\texttt{rec}}
\newcommand{\yes}{\texttt{yes}}
\newcommand{\no}{\texttt{no}}
\newcommand{\mend}{\texttt{end}}
\newcommand{\false}{\texttt{ff}}

\newcommand{\nil}{\texttt{nil}}

\newcommand{\eg}{{\textit{e.g.} }}
\newcommand{\aka}{{\textit{a.k.a.} }}

\newcommand{\trueset}[1]{{\left\llbracket #1 \right\rrbracket}}

\newcommand{\meq}{\sim}
\newcommand{\sys}{\mathit{SYS}}

\begin{document}

\author{Luca Aceto\inst{1} \and Antonis Achilleos\inst{1} \and Adrian Francalanza\inst{2} \and Anna Ing\'{o}lfsd\'{o}ttir\inst{1} \and S\ae{}var \"Orn Kjartansson\inst{1}}

\institute{
	School of Computer Science, Reykjavik University,
	Reykjavik, Iceland
\and 
Dept. of Computer Science, ICT, University of Malta, Msida, Malta  \\
}

\title{Determinizing Monitors for HML with Recursion\thanks{This research was supported by the project ``TheoFoMon: Theoretical Foundations for Monitorability'' (grant number: 163406-051) of the Icelandic Research Fund.}
}





\maketitle

\begin{abstract}
	We examine the determinization of monitors for HML with recursion. We demonstrate that every monitor is equivalent to a deterministic one, which is at most doubly exponential in size with respect to the original monitor. When monitors are described as CCS-like processes, this doubly exponential bound is optimal. When (deterministic) monitors are described as finite automata (as their LTS), then they can be exponentially more succinct than their CCS process form.
\end{abstract}

\section{Introduction}

\emph{Monitors} are computational entities that observe the executions of other computing entities (referred to hereafter  as \emph{systems}, or, in more formal settings, as \emph{processes}) with the aim of 
accruing  system
information \cite{KleinG15,HeCL13}, comparing system 
executions against behavioral specifications \cite{rvpaper}, or reacting to 
observed  executions via 
adaptation 
or 
enforcement procedures \cite{CassarF16,Ligatti05}.
Monitor descriptions can vary substantially, from pseudocode 
\cite{Geilen,Erling04:inlined}, to 
mathematical descriptions \cite{VardiW94,dAmorimR05,FalconeFM12}, to 
executable code  in a domain-specific or general-purpose language \cite{DeboisHS15,MJG+11mop}.     Because they are 
part of the trusted computing base, monitor descriptions are expected to be \emph{correct}, and a 
prevalent 
correctness requirement is that they should exhibit \emph{deterministic behavior}.  

Monitors are central for the field of Runtime Verification, where typically we use monitors to observe the trace produced during the run of a process. The monitor is expected to be able to reach a verdict after reading a finite part of the execution trace, if it can conclude that the monitored process violates (or, dually, satisfies) a certain specification property. Such specification properties are often expressed in an appropriate logical language, such as LTL \cite{pnueli1977temporal,truncated,goodbadugly,rvLTL,comparingLTL}, CTL, CTL$^*$ \cite{clarke1981design}, and $\muhml$ \cite{KOZEN1983333,AceILS:2007,Larsen90,rvpaper}.
For a brief overview of Runtime Verification, see \cite{Leu:RV:Overv}.

In \cite{rvpaper}, Francanalza et al.  studied the monitorability of properties expressed in full $\muhml$. They determined a maximal monitorable fragment of the logic, for which they introduced a monitor generating procedure.
The monitors constructed from this procedure can be nondeterministic, depending on the $\muhml$ formula from which they were constructed.
In the light of the importance of deterministic monitors in Runtime Verification,
we would like to be able to determinize these monitors.

In this paper we tackle the problem of determinizing monitors for $\muhml$ in the framework of \cite{rvpaper}, where monitors are described using syntax close to the regular fragment of CCS processes \cite{Milner:1982:CCS:539036}. 
We demonstrate that every monitor can be transformed into an equivalent deterministic one, but the price can be a hefty one: there are monitors which require a doubly exponential blowup in size to determinize.
 Although we focus on the monitors employed in [16], 
our methods and results can be 
extended to 
other cases of similar monitors, when these monitors are described using syntax that is close to that of the regular fragment of CCS processes. 
Furthermore, we demonstrate that finite automata, as a specification language, can be exponentially more succinct than monitors as we define them, exponentially more succinct  than the monitorable fragment of $\muhml$, and doubly exponentially more succinct than the deterministic monitorable fragment of $\muhml$.


The deterministic behavior of monitors is desirable as a correctness requirement and also due to efficiency concerns. A monitor is expected to not overly affect the observed system and thus each transition of the monitor should be performed as efficiently as possible. For each observed action of a system,  the required transition of a monitor is given explicitly by a deterministic monitor, but for a nondeterministic monitor, we need to keep track of all its possible configurations, which can introduce a significant overhead to the monitored system.

On the other hand, as we argue below, nondeterminism arises naturally in some scenarios in the behavior of monitors. In those cases, the most natural monitor synthesized from the properties of interest are nondeterministic and might need to be determinized in order to increase the efficiency of the monitoring process. This observation motivates our study of monitor determinization and its complexity.

\paragraph{Empirical Evidence for the Nondeterministic Behavior of Monitors:}
\label{sec:evid-from-test}

By most measures,  monitoring  is a relatively new software technique and thus not  very widespread.  Instead, we substantiate our  posit that nondeterministic monitors would occasionally occur (should the technique gain wider acceptance) by 
considering \emph{testing} as a 
related pervasive technique.  Despite intrinsic differences\footnote{Tests are executed pre-deployment, and employ more mechanisms to direct the  execution of the system under scrutiny 
	\eg mocking by inputting specific values.} 
tests share several core characteristics with many monitors: they rely on the execution of a system to infer correctness attributes, they often fall under the responsibility of quality assurance software teams  
and, more crucially, they are also expected to behave deterministically \cite{Zhang:2014,LuoHEM14}.     We 
contend that monitors are generally more complex artefacts than tests, which allows us to 
claim reliably that evidence of nondeterminism found in testing  carries over to monitoring.     When compared to tests, 
monitors typically employ more loop and branching constructs in their control structures  
since monitored runs last longer than test executions.
Moreover,  monitoring is  generally regarded as background processing, 
and this expected passivity forces monitoring code to cater for 
various eventualities, aggregating system execution cases that would otherwise be treated by separate (simpler) tests driving the system.  

In testing, the impact of nondeterministic (\aka \emph{flaky}) tests on software development processes is substantial enough (for instance, as reported in \cite{LuoHEM14}, about $4.56\%$ of test failures of the 
Test Anything
Protocol 
system at Google are caused by flaky tests) to warrant the consideration of various academic studies \cite{MemonC04,MarinescuHC14,LuoHEM14}.  These studies concede that flaky tests are hard to eradicate. Detection is hard, partly because tests are executed in conjunction with the system (where the source of nondeterministic behaviour is harder to delineate), or because nondeterminism occurs sporadically, triggered only by specific system runs.    Once nondeterminism has been detected, its causes may be even harder to diagnose: test runs are not immune to Heisenbugs \cite{Gray86} and problems associated with test dependence are not uncommon \cite{Zhang:2014}, all of which impacts on the capacity to fix the offending tests.   In fact, it is common practice to live with nondeterministic tests by annotating them (\eg \texttt{@RandomFail} in Jenkins), perhaps requiring reruns for these tests (\eg  \texttt{@Repeat} in Spring).  Curiously, studies have shown that  developers are often reluctant to remove detected flaky tests (for instance, by using the  \texttt{@Ignore} tag in JUnit) because they may still reveal  system defects (albeit inconsistently) \cite{LuoHEM14,MemonC04}.  

\paragraph{Overview:}

In Section \ref{sec:background}, we give the necessary background for our investigation. 
We introduce $\muhml$ formulas and our framework for processes, monitors, and finite automata.
In Section \ref{sec:rewriting}, we prove that all monitors for $\muhml$ can be determinized. We provide two methods for the determinization of monitors. 
One reduces the determinization of monitors to the determinization of regular CCS processes, as performed by Rabinovich in \cite{rabinovich}. 
The other directly applies a procedure for transforming systems of equations, inspired from Rabinovich's methods, directly on $\muhml$ formulas to turn them into a deterministic form. Then, using the monitor synthesis procedure from \cite{rvpaper}, one can construct monitors which are immediately deterministic.
In Section \ref{sec:bounds}, we examine the cost of determinizing monitors more closely and we compare the size and behavior of monitors to the size (number of states in this case) and behavior of corresponding finite automata. We examine the simpler case of single-verdict monitors,
namely monitors which are allowed to reach verdict \yes\ or verdict \no, but not both (or to halt without reaching an answer).
Section \ref{sec:quick_fix} explains how to extend the methods and bounds of Section \ref{sec:bounds} from single-verdict monitors to the general case of monitors with multiple verdicts.
The reader is encouraged to see Section \ref{sec:conclusions} for an extensive summary of the technical results in this paper.


%

\section{Background}
\label{sec:background}

We provide background on the main definitions and results for monitoring $\muhml$ formulae, as defined in \cite{rvpaper},  and present the conventions we use in this paper.

\subsection{Basic Definitions: Monitoring $\muhml$ Formulae on Processes}
\label{sec:def}

We begin by giving the basic definitions for the main concepts we use. These include the calculus used to model a system, the logic used to reason about the systems and finally the monitors used to verify whether a system satisfies some specific formula in the logic.

\subsubsection{The Model}

The processes whose properties we monitor (and on which we interpret the $\muhml$ formulae) are states of a labeled transition system (LTS).
%
    A 
    labeled transition system 
    is
    a triple
    \[
        \langle \mycap{Proc}, (\mycap{Act}\cup\{\tau\}),\rightarrow\rangle
    \]
    where $\mycap{Proc}$ is a set of states or processes, $\mycap{Act}$ is a set of observable actions, $\tau \notin \mycap{Act}$ is the distinguished silent action, and $\rightarrow\subseteq(\mycap{Proc}\times(\mycap{Act}\cup\{\tau\})\times \mycap{Proc})$ is a transition relation. 
The syntax of the processes in \mycap{Proc} is defined by the following grammar:
    \begin{align*}
        p,q\in \mycap{Proc} ::&= \texttt{nil}      &&|~~\alpha.p &&|~~p+q
                       &&|~~\texttt{rec} x.p  &&|~~x
    \end{align*}
    where $x$ comes from a countably infinite set of process variables. 
    These processes are a standard variation on the regular fragment of $CCS$ \cite{Milner:1982:CCS:539036};
    in \cite{rabinovich}, these processes are called $\mu$-expressions. We simply call them processes in this paper.
    As usual, the construct $\rec\ x.p$ binds the free occurrences of $x$ in $p$. 
    In what follows, unless stated otherwise, we focus on processes without the occurrences of free variables.
    The substitution operator $p[q/x]$ is defined in the usual way.
%
    The transition relation $\rightarrow$ and thus the behavior of the processes is defined by the  derivation rules in Table \ref{tab:ProcessRules}.
    \begin{table}
    	\begin{align*}
        \mycap{Act}&\frac{}{\alpha.p\xrightarrow{\alpha}p}                       &&&
        \mycap{Rec}&\frac{}{\texttt{rec} x. p\xrightarrow{\tau}p[\texttt{rec} x.p/x]} \\
        \mycap{SelL}&\frac{p\xrightarrow{\mu}p'}{p+q\xrightarrow{\mu}p'}   &&&
         \mycap{SelR}&\frac{q\xrightarrow{\mu}q'}{p+q\xrightarrow{\mu}q'} 
    \end{align*}
    where $\alpha\in\mycap{Act}$ and $\mu\in\mycap{Act}\cup\{\tau\}$.
    \caption{Dynamics of Processes}
    \label{tab:ProcessRules} 
    \end{table}

    For each $p,p'\in\mycap{Proc}$ and $\alpha\in\mycap{Act}$, we use $p\xRightarrow{\alpha}p'$ to mean that $p$ can derive $p'$ using a single $\alpha$ action and any number of $\tau$ actions, $p(\xrightarrow{\tau})^*.\xrightarrow{\alpha}.(\xrightarrow{\tau})^*p'$. For each $p,p'\in\mycap{Proc}$ and trace $t=\alpha_1.\alpha_2.\dots \alpha_r\in\mycap{Act}^*$, we use $p\xRightarrow{t}p'$ to mean $p\xRightarrow{\alpha_1}.\xRightarrow{\alpha_2}\dots\xRightarrow{\alpha_r}p'$ if $t$ is non-empty and $p(\xrightarrow{\tau})^*p'$ if $t$ is the empty trace $\epsilon$.


\subsubsection{The Logic}

We use $\muhml$, the Hennessy-Milner logic with recursion, to describe properties of the processes. 

\begin{definition}
    The formulae of $\muhml$ are  constructed using the following grammar:
    \begin{align*}
        \varphi,\psi\in\muhml         ::&=
        \texttt{tt}                     &&|~~\texttt{ff}        \\
        &|~~\varphi\land\psi            &&|~~\varphi\lor\psi    \\
        &|~~\langle\alpha\rangle\varphi &&|~~[\alpha]\varphi    \\
        &|~~\min X.\varphi              &&|~~\max X.\varphi     \\
        &|~~X                                                   
    \end{align*}
    where $X$ comes from a countably infinite set of logical variables $\mycap{LVar}$.
\end{definition}

    Formulae are evaluated in the context of a labeled transition system and an environment, $\rho:\mycap{LVar}\rightarrow 2^\mycap{Proc}$, which gives values to the logical variables in the formula.
    For an environment $\rho$, variable $X$, and set $S \subseteq \mycap{Proc}$, $\rho[X \mapsto S]$ is the environment which maps $X$ to $S$ and all $Y \neq X$ to $\rho(Y)$.
    The semantics for $\muhml$ formulae is given through a function $\llbracket \cdot \rrbracket$, which, given an environment $\rho$, maps each formula to a set of processes --- namely the processes which satisfy the formula under the assumption that each $X \in \mycap{LVar}$ is satisfied by the processes in $\rho(X)$. $\llbracket \cdot \rrbracket$ is defined as follows:
    \begin{align*}
        \llbracket\texttt{tt}, \rho\rrbracket&\defeq \mycap{Proc} \ \   \text{ and }   \ \ \llbracket\texttt{ff},\rho\rrbracket\defeq\emptyset \\
        \llbracket\varphi_1\land\varphi_2, \rho\rrbracket&\defeq\llbracket\varphi_1,\rho\rrbracket\cap\llbracket\varphi_2,\rho\rrbracket
        \\
        \llbracket\varphi_1\lor\varphi_2, \rho\rrbracket&\defeq\llbracket\varphi_1,\rho\rrbracket\cup\llbracket\varphi_2,\rho\rrbracket 
        \\
        \llbracket[\alpha]\varphi,\rho\rrbracket&\defeq\left\{p~\big|~ \forall q. \ p\xRightarrow{\alpha}q\text{ implies } q\in\llbracket\varphi,\rho\rrbracket\right\}
        \\
        \llbracket\langle\alpha\rangle\varphi,\rho\rrbracket&\defeq\left\{p~\big|~ \exists q. \ p\xRightarrow{\alpha}q\text{ and } q\in\llbracket\varphi,\rho\rrbracket\right\} \\
        \llbracket\max X.\varphi,\rho\rrbracket&\defeq\bigcup\left\{S~\big|~ S \subseteq \llbracket\varphi, \rho[X\mapsto S]\rrbracket\right\} \\
        \llbracket\min X.\varphi,\rho\rrbracket&\defeq\bigcap\left\{S~\big|~ S \supseteq \llbracket\varphi,\rho[X\mapsto S]\rrbracket\right\}
        \\ 
        \llbracket X, \rho\rrbracket&\defeq\rho(X) .
    \end{align*}
    A formula is closed when every occurrence of a variable $X$  is in the scope of recursive operator $\rec\ X$. Note that the environment, $\rho$, has no effect on the semantics of a closed formula. For a closed formula $\varphi$, we often drop the environment from the notation for $\trueset{\cdot}$ and write $\llbracket\varphi\rrbracket$ instead of $\llbracket\varphi,\rho\rrbracket$. Henceforth we work only with closed formulas, unless stated otherwise.
    Formulae $\varphi$ and $\psi$ are (logically) equivalent, written $\varphi \equiv \psi$, if $\trueset{\varphi,\rho} = \trueset{\psi,\rho}$ for every environment $\rho$.


We focus on  $\shml$, the safety fragment of $\muhml$. Both $\shml$ and its dual fragment, $\chml$, are defined by the grammar:
%
	\begin{align*}
	\theta,\vartheta\in\shml    &::= \texttt{tt}    &&|~\texttt{ff}    &&|~[\alpha]\theta              &&|~\theta\land\vartheta    &&|~\max X.\theta &|~X \\
	\pi,\varpi\in\chml          &::= \texttt{tt}    &&|~\texttt{ff}    &&|~\langle\alpha\rangle\pi     &&|~\pi\lor\varpi           &&|~\min X.\pi &|~X 
.
	\end{align*}

\begin{examp}[a formula]
	Formula $\varphi_e = \max\ X.[\alpha] ([\alpha] \false \wedge X) \in \shml$.
Notice that 
\begin{align*}
\varphi_e &\equiv \\
&\equiv 
[\alpha] ([\alpha] \false \wedge \max\ X.[\alpha] ([\alpha] \false \wedge X)) \\ 
&\equiv 
[\alpha] ([\alpha] \false \wedge [\alpha] ([\alpha] \false \wedge \max\ X.[\alpha] ([\alpha] \false \wedge X)))\\
&\equiv 
[\alpha] [\alpha] \false 
.
\end{align*}
\end{examp}

\subsubsection{Monitors}

We now define the notion of a monitor. Just like for processes, we use the definitions  given in \cite{rvpaper}. Monitors are part of an LTS, much like processes.

\begin{definition}
    The syntax of a monitor is identical to that of a  process,  with the exception that the \texttt{nil} process is replaced by verdicts. A verdict can be one of $\texttt{yes}$, $\texttt{no}$ and $\texttt{end}$, which represent acceptance, rejection and termination respectively. A monitor is defined by the following grammar:
    \begin{align*}
        m,n\in\mycap{Mon}   ::&= v                   && |~~\alpha.m      && |~~m+n           &&|~~\texttt{rec} x.m   &&|~~x  \\
        v\in\mycap{Verd}    ::&= \texttt{end}        && |~~\texttt{no}   && |~~\texttt{yes}                                  
    \end{align*}
    where $x$ comes from a countably infinite set of monitor variables.
\end{definition}

    The behavior of a monitor is defined by the  derivation rules of Table \ref{tab:Orules}.
    \begin{table}[h]
    \begin{align*}
        \mycap{mAct}&\frac{}{\alpha.m\xrightarrow{\alpha}m}                             &&&
        \mycap{mRec}&\frac{}{\texttt{rec} x.m\xrightarrow{\tau}m[\texttt{rec} x.m/x]}   \\
        \mycap{mSelL}&\frac{m\xrightarrow{\mu}m'}{m+n\xrightarrow{\mu}m'}               &&&
        \mycap{mSelR}&\frac{n\xrightarrow{\mu}n'}{m+n\xrightarrow{\mu}n'}               \\
        \mycap{mVerd}&\frac{}{v\xrightarrow{\alpha}v}
    \end{align*}
    where $\alpha\in\mycap{Act}$ and $\mu\in\mycap{Act}\cup\{\tau\}$.
    \caption{Monitor Dynamics}
    \label{tab:Orules}
    \end{table}

    For each $m,m'\in\mycap{Mon}$ and $\alpha\in\mycap{Act}$, we use $m\xRightarrow{\alpha}m'$ to mean that $m$ can derive $m'$ using a single $\alpha$ action and any number of $\tau$ actions, $m(\xrightarrow{\tau})^*.\xrightarrow{\alpha}.(\xrightarrow{\tau})^*m'$. For each $m,m'\in\mycap{Mon}$ and trace $t=\alpha_1.\alpha_2.\dots \alpha_r\in\mycap{Act}^*$, we use $m\xRightarrow{t}m'$ to mean $m\xRightarrow{\alpha_1}.\xRightarrow{\alpha_2}\dots\xRightarrow{\alpha_r}m'$ if $t$ is non-empty and $m(\xrightarrow{\tau})^*m'$ if $t$ is the empty trace.
    Rules \mycap{mSelL} and \mycap{mSelR} may be referred to collectively as \mycap{mSel} for convenience.

\begin{examp}[a monitor]
Let $m_e = \rec\ x. \alpha.(\alpha.\no + x)$. Notice the similarities between $m_e$ and $\varphi_e$.
We will be using $m_e$, $\varphi_e$, and related constructs as a running example.
\end{examp}

\subsubsection{Monitored system}
If a monitor $m\in\mycap{Mon}$ is monitoring a process $p\in\mycap{proc}$, then it must mirror every visible action $p$ performs. If $m$ cannot match an action performed by $p$ and it cannot perform and internal action, then $m$ becomes the inconclusive \texttt{end} verdict. We are only looking at the visible actions and so we allow $m$ and $p$ to perform transparent $\tau$ actions independently of each other.

\begin{definition}
    A monitored system is a monitor $m\in\mycap{Mon}$ and a process $p\in\mycap{Proc}$ which are run side-by-side, denoted $m\triangleleft p$. The behavior of a monitored system is defined by the  derivation rules in Table \ref{tab:synthesis}.
    \begin{table}[h]
    	\begin{align*}
        \mycap{iMon}&\frac{p\xrightarrow{\alpha}p'~~~m\xrightarrow{\alpha}m'}{m\triangleleft p\xrightarrow{\alpha}m'\triangleleft p'} &&&
        \mycap{iTer}&\frac{p\xrightarrow{\alpha}p'~~~m\not\xrightarrow{\alpha}~~~m\not\xrightarrow{\tau}}
                          {m\triangleleft p\xrightarrow{\alpha}\texttt{end}\triangleleft p'}                    \\
        \mycap{iAsyP}&\frac{p\xrightarrow{\tau}p'}{m\triangleleft p\xrightarrow{\tau}m\triangleleft p'}    &&&
        \mycap{iAsyM}&\frac{m\xrightarrow{\tau}m'}{m\triangleleft p\xrightarrow{\tau}m'\triangleleft p}    
    \end{align*}
    \caption{Monitored Processes}
    \label{tab:synthesis}
    \end{table}
\end{definition}

If a monitored system $m\triangleleft p$ can derive the \texttt{yes} verdict, we say that $m$ accepts $p$, and similarly $m$ rejects $p$ if the monitored system can derive \texttt{no}.
\begin{definition}[{\bf Acceptance/Rejection}]
    $\acc(m,p)\defeq\exists t,p'.\ m\triangleleft p\xRightarrow{t}\texttt{yes}\triangleleft p'$ and $\rej(m,p)\defeq\exists t,p'.\ m\triangleleft p\xRightarrow{t}\texttt{no}\triangleleft p'$.
\end{definition}

\subsubsection{Finite Automata}

We now present a brief overview of Finite Automata Theory, which we use in Section \ref{sec:bounds}. The interested reader can see \cite{sipser1997introduction} for further details.

A nondeterministic finite automaton (NFA) is a quintuple $A = (Q, \Sigma, \delta, q_0, F)$, where $Q \neq \emptyset$ is a set of states, $\Sigma$ is a set of symbols, called the alphabet --- in our context, $\Sigma = \mycap{Act}$ ---, $\delta \subseteq Q\times \Sigma \times Q$ is a transition relation, $q_0 \in Q$ is the initial state, and $F \subseteq Q$ is the set of final or accepting states. Given a word $t \in \Sigma^*$, a run of $A$ on $t = t_1 \cdots t_k$ is a sequence $q_0 q_1\cdots q_k$, such that for $1 \leq i \leq k$, $(q_{i-1},t_i,q_i) \in \delta$; the run is called accepting if $q_k \in F$. We say that $A$ accepts $t$ when $A$ has an accepting run on $t$. We say that $A$ accepts/recognizes a language $L \subseteq \Sigma^*$ when $A$ accepts exactly the words in $L$; $L$ is then unique and we call it $L(A)$. If $\delta$ is a function $\delta:Q\times \Sigma \rightarrow Q$ (depending on the situation, one may just demand that $\delta$ is a partial function), then $A$ is a deterministic finite automaton (DFA).
It is a classical result that for every NFA with $n$ states, there is an equivalent DFA (i.e. a DFA which recognizes the same language) with at most $2^n$ states --- furthermore, this upper bound is optimal.

\begin{theorem}[\cite{rabin1959finite}]
	If $A$ is an NFA of $n$ states, then there is a DFA of at most $2^n$ states which recognizes $L(A)$.
\end{theorem}

\begin{proof}
	The way to construct such  a DFA is through the classical subset construction.
	For more details, see \cite{sipser1997introduction}, or another introductory text for automata theory or theoretical computer science.
\qed\end{proof}

\paragraph{Are monitors automata?}
The reader may wonder at this point whether a monitor is simply an NFA in disguise. Certainly, the LTS of  monitor resembles an NFA, although there are differences. A monitor --- at lest in this paper --- is identified with its syntactic form, as defined above, although, as Section \ref{sec:bounds} demonstrates, it pays to be a little more relaxed on this constraint. Monitors can also reach both a \yes\ and a \no\ verdict --- and sometimes both for the same trace. Again, this is often undesirable. Another difference is that once a monitor reaches a verdict, there is no way to change its decision by seeing more of the trace: verdicts are irrevocable for monitors. 
For more on the relationship between monitors and automata, the reader should read Section \ref{sec:bounds}.

\subsection{Previous Results}
\label{sec:prev}

The main result from \cite{rvpaper} is to define a subset of $\muhml$ which is monitorable and show that it is maximally expressive. This subset is called $\mhml$ and consists of the safe and co-safe syntactic subsets of $\muhml$: $\mhml\defeq\shml\cup\chml$.
From now on, we focus on \shml, but the case of \chml\ is dual. The interested reader can see \cite{rvpaper} for more details.
%

In order to prove that $\shml$ is monitorable, in \cite{rvpaper} Francalanza, Aceto, and Ing\'{o}lfsd\'{o}ttir define a monitor synthesis function, $\msf{-}$, which maps formulae to monitors, and show that for each $\varphi\in\mhml$, $\msf{\varphi}$ monitors for $\varphi$, in that $\rej(\msf{\varphi}, p)$ exactly for those  processes $p$ for which $p \notin \trueset{\varphi}$.
This function is used in the proofs in Section \ref{sec:rewriting} 
and so we  give the definition here.

\begin{definition}[\bf Monitor Synthesis]
    \label{def:msf}
    \begin{align*}
        \msf{\texttt{tt}}&\defeq\texttt{yes} & \msf{\texttt{ff}}&\defeq\texttt{no} & \msf{X}\defeq x 
    \end{align*}
    \begin{align*}
        \msf{[\alpha]\psi}&\defeq
        \begin{cases}
            \alpha.\msf{\psi}   &\text{if $\msf{\psi}\neq\texttt{yes}$} \\
            \texttt{yes}        &\text{otherwise}                       \\
        \end{cases}
        \\
        \msf{\psi_1\land\psi_2}&\defeq
        \begin{cases}
            \msf{\psi_1}                &\text{if $\msf{\psi_2}=\texttt{yes}$}  \\
            \msf{\psi_2}                &\text{if $\msf{\psi_1}=\texttt{yes}$}  \\
            \msf{\psi_1}+\msf{\psi_2}   &\text{otherwise}                       \\
        \end{cases}
        \\
        \msf{\max X.\psi}&\defeq
        \begin{cases}
            \texttt{rec} x.\msf{\psi}   &\text{if $\msf{\psi}\neq\texttt{yes}$} \\
            \texttt{yes}                &\text{otherwise}                       \\
        \end{cases}
    \end{align*}
\end{definition}

We say that a monitor $m$ monitors for a formula $\varphi \in \shml$ when for each process $p$, $\rej(m, p)$ if and only if $p \notin \trueset{\varphi}$.

\begin{theorem}[{\bf Monitorability}, \cite{rvpaper}]
    For each $\varphi\in\shml$,  $\msf{\varphi}$ monitors for $\varphi$.
\end{theorem}

\begin{examp}
	Notice that $\msf{\varphi_e} = m_e$. On the other hand, we also know that for $\varphi'_e = [\alpha][\alpha]\false$, $\varphi_e \equiv \varphi'_e$.
Therefore, for $m'_e = \msf{\varphi'_e} = \alpha.\alpha.\no$, $m_e$ and $m'_e$ monitor for the same formula, and therefore they are equivalent.
\end{examp}

\subsection{Determinism,  Verdicts, and the Choices That We Make}
\label{subsec:det_verd_choices}

The purpose of this paper is to examine the determinization of monitors, which is the process of constructing a deterministic monitor from an equivalent nondeterministic one. Therefore, we must know what a deterministic monitor is and what it  means that two monitors are equivalent.

Even before having defined what a deterministic monitor is, it is easy to find examples of nondeterministic ones. One can look in \cite{rvpaper} for several examples. An easy one is $m_c = \alpha.\yes + \alpha.\no$. This monitor can reach a positive \emph{and} a negative verdict  for each trace which starts with $\alpha$, but has no opinion on any other trace. There are two ways to amend this situation. One is to make sure that different actions can transition to different monitors, as in $\alpha.\yes + \beta.\no$, thus removing a nondeterministic choice; the other is to consider single-verdict monitors which can reach only one verdict, like $\alpha.\yes + \alpha.p$. We explore both approaches.

\subsubsection{Determinism}

For Turing machines, algorithms, and finite automata, determinism means that from every state (in our case, monitor) and input symbol (in our case, action), there is a unique transition to follow. In the case of monitors, we can transition either through an action from \mycap{Act}, but also through a $\tau$-action, which can occur without reading from a trace. In the context of finite automata, these actions would perhaps correspond to $\epsilon$-transitions, which are eliminated from deterministic automata. We cannot eliminate $\tau$-transitions from deterministic monitors, because we need to be able to activate the recursive operators. What we can do is to make sure that there is never a choice between a $\tau$-transition and any other transition.

In other words, in order for a monitor to run deterministically, it cannot contain a nondeterministic choice between two sub-monitors reached by the same action, $m\xrightarrow{\alpha}n_1$ and $m\xrightarrow{\alpha}n_2$, or a nondeterministic choice between performing an action or performing the transparent action, $m\xrightarrow{\alpha}n_1$ and $m\xrightarrow{\tau}n_2$. It must also be impossible to derive these choices using the derivation rules above.
As we can see from Table \ref{tab:Orules}, such choices can only be introduced by sums --- that is, monitors of the form $m_1+m_2$. Note that a sum $m_1+m_2+\cdots+m_k$ can also be written as $\sum_{i=1}^k m_i$.

\begin{definition}
	A monitor $m$ is deterministic iff
	every sum of at least two summands
		which appears in $m$ 
	is of the form 
	$\sum_{\alpha \in A} \alpha.m_\alpha$, 
	where $A \subseteq \mycap{Act}$.
\end{definition}
Note that under this definition, from every deterministic monitor $m$ and action $\alpha \in \mycap{Act}$, if  $m \xRightarrow{\alpha} m'$, then all derivations $m \xRightarrow{\alpha}$  begin with a unique transition.
As we will see in the following  sections, 
this set of monitors is in fact maximally expressive.

\begin{examp}
	Notice that monitor $m_e$ is not deterministic, but it has an equivalent deterministic monitor, which is $m'_e$.
\end{examp}

\begin{examp}[a simple server (from \cite{rvpaper})]
Assume that we have a very simple web server $p$ that alternates between accepting a request from a client, \texttt{req}, and sending a response back, \texttt{res}, until the server terminates, \texttt{cls}. We want to verify that the server cannot terminate while in the middle of serving a client (after executing \texttt{req} but before executing \texttt{res}). We can encode this property in the following $\shml$ formula
\[
\varphi = \max X.([\texttt{req}][\texttt{cls}]\texttt{ff} \land [\texttt{req}][\texttt{res}]X)
\]
We can then define a violation complete monitor for this formula using the monitor synthesis function:
\[
m = \texttt{rec} x.(\texttt{req.cls.no} + \texttt{req.res.}x)
\]
Since $m$ contains a choice between \texttt{req.cls.ff} and $\texttt{req.res.}x$, it is nondeterministic. However it is possible to define a deterministic monitor that monitors for $\varphi$. For example:
\[
m' = \texttt{req}.(\texttt{res.rec} x.\texttt{req}.(\texttt{res}.x + \texttt{cls.no}) + \texttt{cls.no})
\]
\end{examp}

In general, given a monitor produced by the monitor synthesis function, we can define a deterministic monitor that is equivalent to it.

\begin{theorem}
	For each formula $\varphi\in\mhml$, there exists a deterministic monitor, $m\in\mycap{Mon}$, such that $m$ monitors for $\varphi$.
\end{theorem}

Section \ref{sec:rewriting} is devoted to providing two proofs for this theorem. The first one is presented in Subsection \ref{sect:mon_rewrite} and shows that for each valid monitor, there exists a deterministic monitor that is equivalent to it. The second proof is presented in Subsection \ref{sect:form_rewrite} and shows that there is a subset of $\shml$ for which the monitor synthesis function will always produce a deterministic monitor and that this subset is a maximally expressive subset of $\shml$. The first proof directly depends on Rabinovich's determinization of processes in \cite{rabinovich}, while the second one is inspired by Rabinovich's construction.

\subsubsection{Multiple Verdicts} 

As Francalanza et al. demonstrated in \cite{rvpaper}, to monitor for formulae of $\muhml$, it is enough (and makes more sense) to consider single-verdict monitors. These are monitors, which can use verdict \yes\ or \no, but not both. As we are interested in determinizing monitors to use them for monitoring $\mhml$ properties, confining ourselves to single-verdict monitors is reasonable. 
The constructions of Section \ref{sec:rewriting} are independent of such a choice and are presented for all kinds of monitors, but in Section \ref{sec:bounds} it is more convenient 
 to consider only single-verdict monitors.
Of course, monitor determinization is an interesting problem in its own right and monitors may be useful in other situations besides just monitoring for \mhml. 
We still confine ourselves to determinizing single-verdict monitors, but we present a straightforward approach for dealing with monitors which use both verdicts in Section \ref{sec:quick_fix}. 
In Section \ref{sec:bounds}, whenever we say monitor, we will mean single-verdict monitor. Specifically, 
in Section \ref{sec:bounds}, we assume that the verdict is \yes\ (and \mend, which is often omitted), because we compare monitors to automata and a monitor reaching a \yes\ verdict on a trace intuitively corresponds to an automaton accepting the trace.

\subsubsection{Other Conventions and Definitions}

We will call two monitors, $p$ and $q$ equivalent and write $p \meq q$ if for every trace $t$ and value $v$, $p \xRightarrow{t} v$ iff $q \xRightarrow{t} v$.
Given a monitor, the set of processes it accepts, and thus the set of formulae it monitors, is completely determined by the traces it accepts and rejects (see also \cite{rvpaper}, Proposition 1). Therefore, we compare two monitors with respect to relation $\sim$. 
We assume that the set of actions, \mycap{Act}, is finite and, specifically, of constant size.

We call derivations of the form $p \Rightarrow p$ trivial. If $t = t_1 t_2 \in \mycap{Act}^*$, then $t_1$ is an initial subtrace or prefix of $t$ and we write $t_1 \sqsubseteq t$.
We define a sum of $q$ as follows: $q$ is a sum of $q$ and if $r$ is a sum of $q$, then so are $r+r'$ and $r'+r$. 
We say that a sum $s$ of $q$ is acting as $q$ in a derivation $d: s \xRightarrow{t} r$ if in the same derivation we can replace $s$ by $q$ without consequence (i.e. the first step uses rule \mycap{mSlet} to use a derivation starting from $q$).

The size $|p|$ of a monitor $p$ is the size of its syntactic description, as given in Section \ref{sec:def}, defined recursively: 
$|x| = |v| = 1$; $|a.p| = |p| + 1$; $|p+q| = |p| + |q| + 1$; and $|\rec\ x.p| = |p| +1$.
Notice that $|p|$ also coincides with the total number of submonitor occurrences --- 
namely, symbols in $p$.
We also define the height 
of a monitor, $h(p)$, in the following way: $h(v) = h(x) = 1$; $h(p+q) = \max\{h(p),h(q)\}$; $h(\alpha.p) = h(p) + 1$; $h(\rec\ x.p) = h(p)$. 

We assume that all sums of a verdict $v$ are $v$. Indeed, Lemma \ref{lem:no_sums_of_verdicts} below demonstrates that all instances of $p + v$ can be replaced by $p + \sum_{\alpha \in \mycap{Act}} \alpha.v$.

\begin{lemma}\label{lem:no_sums_of_verdicts}
	Monitor $p + v$ is equivalent to $p + \sum_{\alpha \in \mycap{Act}} \alpha.v$.
\end{lemma} 

\begin{proof}
	Let $\mu \in \mycap{Act} \cup \{\tau \}$ and $q$ a monitor. Then, $p + v \xrightarrow{\mu} q$ iff $p \xrightarrow{\mu} q$ or $v \xrightarrow{\mu} q$.
	But $v \xrightarrow{\mu} q$ exactly when $q = v$ and $\mu \in \mycap{Act}$, therefore 
	$v \xrightarrow{\mu} q$ exactly when $\sum_{\alpha \in \mycap{Act}} \alpha.v \xrightarrow{\mu} q$. In conclusion,  $p + v \xrightarrow{\mu} q$ iff $p + \sum_{\alpha \in \mycap{Act}} \alpha.v \xrightarrow{\mu} q$.
\qed\end{proof}

We also assume that there is no (sub)monitor of the form $\rec\ x.v$, where $v$ a verdict; otherwise, these can be replaced by simply $v$.  Notice then that there is no $\tau$-transition to a verdict. Furthermore, we assume that a verdict appears in all monitors we consider --- otherwise the monitor does not accept/reject anything and is equivalent to \mend. We say that a monitor $p$ accepts/recognizes a language (set of traces) $L \subseteq \mycap{Act}^*$ when for every trace $t \in \mycap{Act}^*$, $t \in L$ iff $p \xRightarrow{t} \yes$.

\section{Rewriting Methods for Determinization}
\label{sec:rewriting}

In this section we present two methods for constructing a deterministic monitor. The first method, presented in subsection \ref{sect:mon_rewrite}, uses a result by Rabinovich \cite{rabinovich} for the determinization of processes. 
The second method, presented in subsection \ref{sect:form_rewrite}, uses methods similar to Rabinovich's directly on formulas of \shml\ to transform them to a deterministic form, so that when we apply the monitor synthesis function from \cite{rvpaper} (Definition \ref{def:msf} of Section \ref{sec:def} in this paper), the result is a deterministic monitor for the original formula.

\subsection{Monitor Rewriting}
\label{sect:mon_rewrite}

We show that, for each monitor, there exists an equivalent monitor that is deterministic. Thus, given a formula $\varphi\in\mhml$, we can apply the monitor synthesis function to produce a monitor for $\varphi$ and subsequently rewrite the monitor so that it runs deterministically.

\begin{definition}
	For a process $p$, $Trace(p) = \{t \in \mycap{Act}^* \mid \exists q.\ p \xRightarrow{t} q \}$.
We call two processes $p,q$ trace equivalent when $Trace(p) = Trace(q)$.
\end{definition}

We use a result given by Rabinovich in \cite{rabinovich}, which states that each  process is 
equivalent to a deterministic one, with respect to trace equivalence --- Rabinovich's definition of determinism for processes coincides with ours when a process has no free variables.

\begin{definition}
    Let $D:\mycap{Proc}\rightarrow\mycap{Proc}$ be a mapping from a 
    process (a process, as defined in Subsection \ref{sec:def})
    to a trace equivalent deterministic 
    process.
    Such a mapping exists as shown by Rabinovich in \cite{rabinovich}.
\end{definition}

In order to apply the mapping $D$ to the monitors, we must define a way to encode them as 
processes.
We do this by replacing each verdict $v$ with a process $v.\texttt{nil}$, where $v$ is treated as an action.

\begin{definition}
    We define a mapping $\pi:\mycap{Mon}\rightarrow\mycap{Proc}$ as follows:
    \begin{align*}
        \pi(\alpha.m)               &= \alpha.\pi(m)                \\
        \pi(m + n)                  &= \pi(m) + \pi(n)              \\
        \pi(\texttt{rec} x.m)       &= \texttt{rec} x.\pi(m)        \\
        \pi(v)                      &= v.\texttt{nil}               \\
        \pi(x)                      &= x                            
    \end{align*}
    We can see that $\pi$ 
    maps
    monitors to
    processes where each verdict action $v$ must be followed by the \texttt{nil} process and the \texttt{nil} process must be prefixed by a verdict action. We define an ``inverse'' of $\pi$, which we call $\pi^{-1}$:
    $\pi^{-1}(s) = v$ when $s$ is a sum of some $v.p$ and in all other cases,
    \begin{align*}
        \pi^{-1}(\alpha.p)          &= \alpha.\pi^{-1}(p)           \\
        \pi^{-1}(p + q)             &= \pi^{-1}(p) + \pi^{-1}(q)    \\
        \pi^{-1}(\texttt{rec} x.p)  &= \texttt{rec} x.\pi^{-1}(p)   \\
        \pi^{-1}(x)                 &= x                            
    \end{align*}
\end{definition}

Now if two monitor encodings are trace equivalent, we want the monitors they encode to be equivalent. Remember that we have assumed that all sums of verdict $v$ are simply $v$ (see Lemma \ref{lem:no_sums_of_verdicts}).

\begin{lemma}\label{lem:substitute_v}
	Let $p$ be a monitor and $q = \pi(p)$.
	Then, $p \xRightarrow{t} v$ if and only if there is some $t' \sqsubseteq t$, such that $q \xRightarrow{t'} v.\texttt{nil}$.
\end{lemma}

\begin{proof}
	Straightforward induction on the length of the derivation $q \xRightarrow{t'} v.\texttt{nil}$ for the ``if'' direction and on the length of the derivation $p \xRightarrow{t} v$ for the ``only if'' direction.
\qed\end{proof}

\begin{lemma}\label{lem:substitute_v_nil}
	Let $q = \pi^{-1}(p)$.
	Then, for each $t \in \mycap{Act}^*$, there are some $t' \sqsubseteq t$ and a sum $s$ of some $v.r$, such that $p \xRightarrow{t'} s$, if and only if $q \xRightarrow{t} v$.
\end{lemma}

\begin{proof}
	Like for Lemma \ref{lem:substitute_v}, by induction on the length of the derivations.
\qed\end{proof}

\begin{lemma}
    \label{lemma:trace_equiv}
    Let $m$ and $n = \pi^{-1}(n')$ be monitors such that $Trace(\pi(m)) = Trace(n')$. Then $m \sim n$.
\end{lemma}
    \begin{proof}
        Assume that for a trace $t$ and a verdict $v$, we have $m\xRightarrow{t}v$. By Lemma \ref{lem:substitute_v}, for some $t' \sqsubseteq t$, $\pi(m)\xRightarrow{t'}v.\texttt{nil}$ and so $t'.v\in Trace(\pi(m))$. Since $\pi(m)$ and $n'$ are trace equivalent, we have $n'\xRightarrow{t'v}r$ for some $r$; therefore, $n' \xRightarrow{t'} s$, where $s$ is a sum of some $v.r$, and so, by Lemma \ref{lem:substitute_v_nil},
         $n\xRightarrow{t}v$.
        
        If $n\xRightarrow{t}v$ for some verdict $v$ and trace $t$, then by Lemma \ref{lem:substitute_v_nil}, 
        $n'\xRightarrow{t'} s$, where $s$ is a sum of $v.p$ and $t' \sqsubseteq t$, so $n' \xRightarrow{t'v} p$.
        Then, $\pi(m) \xRightarrow{t'v} p'$, so $p' = \texttt{nil}$. Since all sums of $v$ are $v$ in $m$, it is also the case that all sums of $v.\texttt{nil}$ are $v.\texttt{nil}$ in $\pi(m)$; by Lemma \ref{lem:substitute_v} and rule \mycap{mVerd},    
        $m\xRightarrow{t}v$.
   \qed\end{proof}

\begin{theorem}[\bf Monitor Rewriting]
    For each monitor $m\in \mycap{Mon}$ 
    there exists a deterministic monitor 
	$n \sim m$.
\end{theorem}
    \begin{proof}
        We define a new monitor $n=\pi^{-1}\circ D\circ\pi(m)$. 
        By Lemma \ref{lemma:trace_equiv}, $m$ and $n$ are equivalent. Let $n = \pi^{-1}(n')$.
        Monitor $n$ is deterministic. To prove this, let $s$ be a sum in $n$. We know that $s = \pi^{-1}(s')$, where $s'$ appears in $n'$, which is deterministic. Then, either $s'$ is a sum of a  $v.p$, so $s = v$, or $s' = \sum_{\alpha \in A} \alpha.p'_\alpha$ and $s = \sum_{\alpha \in A} \alpha.p_\alpha$.
   \qed\end{proof}

\begin{examp}
	To determinize $m_e = \rec\ X. \alpha.(\alpha.\no + X)$, we convert it to $p = \pi(m_e) = \rec\ X.\alpha.(\alpha.[no].\nil + X)$. Then, we use Rabinovich's construction to determinize $p$ into (for example) $p' = \alpha.\alpha.[no].\nil$. We can now see that $\pi^{-1}(p') = m'_e$.
\end{examp}

\subsection{Formula Rewriting}
\label{sect:form_rewrite}

In this second approach, we show that each formula $\varphi\in\shml$ is equivalent to some formula in deterministic form which will yield a deterministic monitor if we apply the monitor synthesis function to it. We focus on formulae in $\shml$ but the proof for $\chml$ is completely analogous. We work through an equivalent representation of formulae as systems of equations. The reader can also see \cite{rabinovich}, where Rabinovich uses very similar constructions to determinize processes.

\subsubsection{Systems of Equations}

We give the necessary definitions and facts about systems of equations for $\shml$. These definitions and lemmata are simplified versions of more general constructions. We state necessary lemmata without further explanation, but the reader can see an appropriate source on fixed points and the $\mu$-calculus (for example, \cite{arnold2001rudiments}).

Given tuples, $a = (a_1,\ldots, a_k)$ and $ b = (b_1,\ldots, b_l)$, we use the notation $a \cdot b$ to mean $(a_1,\ldots, a_k,b_1,\ldots, b_l)$.
We abuse notation and extend the operations $\bigcup$ and $\bigcap$ to tuples of sets, so that $\bigcup (a,b) = a \cup b$ and if $a \cdot B \in S^{k+2}$, $\bigcup (a \cdot B) = a \cup \bigcup B$; similarly for $\bigcap$. Also, for tuples of sets, $a = (a_1,\ldots, a_k)$ and $ b = (b_1,\ldots, b_k)$, $a \subseteq b$ iff for all $1 \leq i \leq k$, $a_i \subseteq b_i$.
For an environment $\rho$, a tuple of variables $\mathbb{X} = (X_1,\ldots, X_k)$ and a tuple of sets $ \mathbb{S} = (S_1,\ldots, S_k)$, where $k>1$, we define
\[\rho[X_1 \mapsto S_1,X_2 \mapsto S_2,\ldots X_k \mapsto S_k] = (\rho[X_1 \mapsto S_1])[X_2 \mapsto S_2,\ldots X_k \mapsto S_k]\]
and \[ \rho[\mathbb{X} \mapsto \mathbb{S}] = \rho[X_1 \mapsto S_1,X_2 \mapsto S_2,\ldots X_k \mapsto S_k]. \]
Finally, for an environment $\rho$ and formulae $\varphi_1,\ldots,\varphi_k$, $$\trueset{\bigtimes_{i=1}^n \varphi_i,\rho} = \bigtimes_{i=1}^n\trueset{\varphi_i,\rho},$$
where $\bigtimes_{i=1}^n S_i$ is the $n$-ary cartesian product $S_1 \times S_2 \times \cdots \times S_n$.

\begin{definition}
	A system of equations is a triple $\mathit{SYS}=(Eq, X, \mathcal Y)$ where $X$ is called the principal variable in $\mathit{SYS}$, $\mathcal Y$ is a finite set of variables and $Eq$ is an $n$-tuple of equations:
	\begin{align*}
	X_1 &= F_1,   \\
	X_2 &= F_2,   \\
	&~ \vdots \\
	X_n &= F_n   ,
	\end{align*}
%
	where for $1 \leq i < j \leq n$, $X_i$ is different from $X_j$, $F_i$ is an expression in $\shml$ which can contain variables in $\mathcal Y\cup\{X_1, X_2, \dots, X_n\}$ and there is some $1 \leq i \leq n$, such that $X = X_i$. $\mathcal{Y}$ is called the set of free variables of $\sys$ and is disjoint from $\{X_1, X_2, \dots, X_n\}$.
\end{definition}

As we see further below, a system of equations can alternatively be understood as a simultaneous fixed point, but we mostly use the following recursive definition to provide semantics,
 where
$Sol(\sys,\rho) = ( S_1(\rho),\ldots,S_n(\rho))$ is an $n$-tuple of sets of processes from a labeled transition system, giving solutions to every variable $X_i$ of the equation system, given an environment $\rho$. We use the notations 
$$\rho[\sys]\ \  \defeq\ \  \overline{\rho}^\sys\ \  \defeq\ \  \rho\left[ \bigtimes_{i=1}^n X_i \mapsto Sol(\sys,\rho)\right],$$
depending on which is clearer in each situation. 
For $n = 1$, let
$$Sol(\mathit{SYS},\rho) = 
\trueset{\max X_1.F_1,\rho}.$$ 
 Let $\mathit{SYS}'$ be a system $\mathit{SYS}$ after adding a new first equation, $X = F$ and removing $X$ from $\mathcal Y$; then,
for $Sol(\mathit{SYS},\rho) =( S_1(\rho),\ldots,S_n(\rho))$, 
let for every environment $\rho$, 
$$S_0(\rho) = \bigcup \left\{ S \mid S \subseteq \trueset{ F,\overline{\rho[X \mapsto S]}^\sys } \right\}$$
and 
$$Sol(\mathit{SYS}',\rho) = \left(S_0(\rho),S_1(\rho\left[X \mapsto S_0(\rho)\right]),\ldots,S_n(\rho\left[X \mapsto S_0(\rho)\right])\right).$$

If the primary variable of a system of equations, $\sys$, is $X_i$ and 
$Sol(\mathit{SYS},\rho) = (S_1(\rho),\ldots,S_n(\rho))$, then 
$\trueset{\mathit{SYS}, \rho} = S_i(\rho)$. 
We say that a system of equations $\mathit{SYS}$ is equivalent to a formula $\varphi$ of $\shml$ with free variables from $\mathcal{Y}$ when for every environment $\rho$,
$\llbracket \varphi,\rho \rrbracket = \llbracket \mathit{SYS},\rho \rrbracket$. 

\begin{examp}
	A system of equations 
	equivalent to $\varphi_e = \max\ X.[\alpha]([\alpha]\false \wedge X)$ is $\sys_e$, which has no free variables and includes the equations $X_0 = [\alpha] X_1$ and $X_1 = [\alpha] \false \wedge X_0$, where $X_0$ is the principal variable of $\sys_e$.
\end{examp}

Given an environment $\rho$ and such a system of equations $\mathit{SYS}$, 
Notice that for any equation $X = F$ of $\mathit{SYS}$, $\trueset{X,\rho[\mathit{SYS}]} = \trueset{F,\rho[\mathit{SYS}]}$ and that if $X$ is the principal variable of $\sys$, then $\trueset{X,\rho[\mathit{SYS}]} =  \trueset{\mathit{SYS},\rho}$.
We note that, as is well known, the order of the calculation of the fixed points does not affect the overall outcome: 
\begin{lemma}\label{lem:systems:reorder}
	Let $\sys = (Eq_1,X,\mathcal{Y})$ and $\sys' = (Eq_2,X,\mathcal{Y})$, where $Eq_2$ is a permutation of $Eq_1$. Then, for all environments $\rho$, $\trueset{\sys,\rho} = \trueset{\sys',\rho}$.
\end{lemma}
Therefore, for every equation $X = F$ of the system $\mathit{SYS}$, if $\sys'$ is the result of removing $X=F$ from the equations  and adding $X$ to the free variables,
\[ \trueset{X,\rho[\mathit{SYS}]} = \bigcup \left\{ S \mid S \subseteq \trueset{ F,\overline{\rho[X \mapsto S]}^{\sys'} } \right\}. \]

Furthermore, we can compute parts of the solution of the system (or the whole solution) as simultaneous fixed points:
 \begin{lemma}\label{lem:systems:simultaneous}
 	Let $\sys = (Eq_1\cdot Eq_2,X,\mathcal{Y})$, $Eq_1 = \bigtimes_{1 \leq i \leq k} \{X_i = F_i\}$, $\mathcal{X}_1 = (X_1,\ldots,X_k)$, and $\sys' = (Eq_2,X,\mathcal{Y}\cup \mathcal{X}_1)$. Let for all environments $\rho$, 
 	\[
 	\mathbb{S}_0(\rho) = \bigcup \left\{ \mathbb{S} \mid \mathbb{S} \subseteq \trueset{ \bigtimes_{i=1}^k F_i,\overline{\rho[\bigtimes\mathcal{X}_1 \mapsto \mathbb{S}]}^{\sys'} } \right\}
 	.\]
 	Then, 
 	$$
 	Sol(\sys,\rho) = \mathbb{S}_0(\rho) \cdot 
			Sol(\sys',\rho[\bigtimes\mathcal{X}_1 \mapsto \mathbb{S}_0(\rho)])	 	
 	.
 	$$
 \end{lemma}
As a consequence, for a set of variables $\{ X_i \mid i \in I \}$ of a system of equations $\sys$,
\[
\trueset{ \bigtimes_{i \in I}X_i ,\rho[\sys]}  
=
 \bigcup \left\{ \mathbb{S} \mid \mathbb{S} \subseteq  \trueset{ \bigtimes_{i \in I}F_i ,\overline{\rho\left[\bigtimes_{i \in I}X_i \mapsto \mathbb{S}\right]}^{\sys'}} \right\},
\]
where $\sys'$ is the result of  removing $X_i = F_i$ from the equations and adding $X_i$ to the free variables, for all $i \in I$.

\subsubsection{Standard and Deterministic Forms}

We begin by defining a deterministic form for formulae in $\shml$.

\begin{definition}
    A formula $\varphi\in\shml$ is in deterministic form iff
	for each pair of formulae $\psi_1 \neq \psi_2$ that occur in the same conjunction in $\varphi$, it must be the case that $ \psi_1 = [\alpha_1]\psi'_1$ and $\psi_2 = [\alpha_2]\psi'_2$ for some $\alpha_1\neq\alpha_2$ or that one of them is a free variable of $\varphi$.
\end{definition}

The following lemma justifies calling these formulae deterministic by showing that applying the monitor synthesis function to them will yield a deterministic monitor.

\begin{lemma}
    Let $\varphi\in\shml$ be a formula deterministic form with no free variables. Then $m=\llparenthesis\varphi\rrparenthesis$ is deterministic.

    \begin{proof}
        By examining the definition of the monitor synthesis function, we can see that if $\llparenthesis\varphi\rrparenthesis$ 
        contains a sum $\sum_{i \in A} p_i$, then $\varphi$ contains a conjunction $\bigwedge_{i \in A \cup B} \varphi_i$, where for  $i \in A$, $\llparenthesis\varphi_i\rrparenthesis = p_i \neq \yes$. 
   \qed\end{proof}
\end{lemma}

\begin{examp}
	$\varphi_e$ is not in deterministic form, but $\psi_e = [\alpha]([\alpha] \false \wedge X)$ is (because here $X$ is free) and so is $\varphi'_e = [\alpha][\alpha]\false$.
\end{examp}

We also define a standard form for formulae in $\shml$. 

\begin{definition}
    A formula $\varphi\in\shml$ is in standard form if all free and unguarded variables in $\varphi$ are at the top level; that is,
    \[
        \varphi = \varphi'\land\bigwedge_{i\in S}X_i
    \]
    where $\varphi'$ does not contain a free and unguarded variable.
\end{definition}

\begin{examp}
	 $\varphi_e$ is in standard form and so is $\psi_e = [\alpha]([\alpha] \false \wedge X)$  (although here $X$ is free, it is also guarded). Formula $[\alpha] \false \wedge X$ is in standard form, because $X$ is at the top level and so is $\varphi'_e = [\alpha][\alpha]\false$, because it has no variables.
	 Formula  $\max\ X.([\alpha]X \wedge Y)$ is not is in standard form, but the construction of Lemma \ref{lemma:top_lvl_var} transforms it into $(\max\ X.[\alpha]X) \wedge Y$ which is.
\end{examp}

\begin{lemma}
    \label{lemma:top_lvl_var}
    Each formula in $\shml$ is equivalent to some formula in $\shml$ which is in standard form.
\end{lemma}
    \begin{proof}
        We define a function $f$ as follows:
        \[
            f(\varphi) = \{i~|~\text{$X_i$ occurs free and unguarded in $\varphi$}\}
        \]
        where $X_1, X_2, \dots$ are all the variables that can occur in the formulae.

        Then formally our claim is that for each $\varphi\in\shml$, there exists a formula, $\psi\in\shml$ such that
        \[
            \varphi\equiv\psi\land\bigwedge_{i\in f(\varphi)}X_i
        \]
        where $f(\psi) = \emptyset$.

        We use induction on the size of $\varphi$ to prove this claim and go through each case below.

        \begin{description}
        \item[$\varphi\in\{\texttt{tt}, \texttt{ff}\}$:]
            This case holds trivially since $f(\varphi)=\emptyset$ and
            \[
                \varphi\equiv\varphi\land\bigwedge_{i\in\emptyset}X_i.
            \]
        \item[$\varphi = X_j$:]
            This case holds trivially since $f(\varphi) = \{j\}$ and
            \[
                \varphi\equiv\texttt{tt}\land\bigwedge_{i\in\{j\}}X_i.
            \]
        \item[{$\varphi = [\alpha]\varphi'$}:]
            Since $\varphi$ is prefixed with the $[\alpha]$ operator, all variables are guarded in $\varphi$ and so $f(\varphi) = \emptyset$ and 
            \[
                \varphi\equiv\varphi\land\bigwedge_{i\in\emptyset}X_i.
            \]
        \item[$\varphi = \varphi_1\land\varphi_2$:]
            By the induction hypothesis, there exist formulae $\psi_1,\psi_2\in\shml$ such that
            \begin{align*}
                f(\psi_1) &= f(\psi_2) = \emptyset ,\\
                \varphi_1 &\equiv \psi_1\land\bigwedge_{i\in f(\varphi_1)}X_i , \\
                \text{ and }\
                \varphi_2 &\equiv \psi_2\land\bigwedge_{i\in f(\varphi_2)}X_i .
            \end{align*}
            Using the fact that $\varphi\land \varphi\equiv \varphi$ for each formula $\varphi\in\muhml$, we have
            \begin{align*}
                \varphi &\equiv \varphi_1\land\varphi_2 \\
                        &\equiv \left(\psi_1\land\bigwedge_{i\in f(\varphi_1)}X_i\right)
                          \land \left(\psi_2\land\bigwedge_{i\in f(\varphi_2)}X_i\right) \\
                        &\equiv (\psi_1\land\psi_2)\land \bigwedge_{i\in f(\varphi_1)}X_i
                                                   \land \bigwedge_{i\in f(\varphi_2)}X_i \\
                        &\equiv (\psi_1\land\psi_2)\land \bigwedge_{i\in f(\varphi_1)\cup f(\varphi_2)}X_i
            \end{align*}
            and since $f(\psi_1) = f(\psi_2) = \emptyset$, we have $f(\psi_1\land\psi_2) = \emptyset$.

            Each free and unguarded variable in $\varphi$ must either be free and unguarded in $\varphi_1$ or $\varphi_2$ and each such variable in $\varphi_1$ or $\varphi_2$ must also be free and unguarded in $\varphi$. This gives us $f(\varphi_1)\cup f(\varphi_2) = f(\varphi)$ and so we have
            \[
                \varphi \equiv (\psi_1\land\psi_2)\land \bigwedge_{i\in f(\varphi)}X_i.
            \]

        \item[$\varphi = \max X_j.\varphi'$:]
            By the induction hypothesis, there exists a formula $\psi\in\shml$ such that
            \[
                \varphi' \equiv \psi\land\bigwedge_{i\in f(\varphi')}X_i
            \]
            where $f(\psi) = \emptyset$.
            
            We use the following equivalences:
            \begin{align}
                \max X.\vartheta &\equiv \vartheta[\max X.\theta/X] \\
                \max X.(\vartheta\land X) &\equiv \max X.\vartheta \\
                Y[\vartheta/X] &\equiv Y \text{~~~~~where $X\neq Y$}
            \end{align}

            From this we get:
            \begin{align*}
                \varphi
                &\equiv\max X_j.\varphi' \\
                &\equiv\max X_j.\left(\psi\land\bigwedge_{i\in f(\varphi')\setminus\{j\}}X_i\right) \\
                &\equiv\left(\psi\land\bigwedge_{i\in f(\varphi')\setminus\{j\}}X_i\right)\left[\max X_j.\left(\psi\land{\bigwedge}_{i\in f(\varphi')\setminus\{j\}}X_i\right)\Big/X_j\right] \\
                &\equiv\psi\left[\max X_j.\left(\psi\land{\bigwedge}_{i\in f(\varphi')\setminus\{j\}}X_i\right)\Big/X_j\right]\land\bigwedge_{i\in f(\varphi')\setminus\{j\}}X_i .
            \end{align*}

            Since each variable in $\psi$ is guarded, substituting a variable for a formula will not introduce unguarded variables and so
            \[
                f\left(\psi\left[\max X_j.\left(\psi\land{\bigwedge}_{i\in f(\varphi')\setminus\{j\}}X_i\right)\Big/X_j\right]\right) = \emptyset.
            \]

            The variables in $\varphi$ that are free and unguarded are exactly the ones that are free and unguarded in $\varphi'$, excluding $X_j$ and so we have
            \[
                f(\varphi) = f(\varphi')\setminus\{j\}.
            \]

            This gives us:
            \[
                \varphi \equiv \psi\left[\max X_j.\left(\psi\land{\bigwedge}_{i\in f(\varphi')\setminus\{j\}}X_i\right)\Big/X_j\right]\land\bigwedge_{i\in f(\varphi)}X_i.
            \]
            \qed
        \end{description}
    \end{proof}

\begin{lemma}\label{lem:replacement_lemma}
	Let $\mathit{SYS}$ be a system of equations and $X = F$ an equation of the system and let $\mathit{SYS}'$ result from replacing $X = F$ by $X = F'$.
	Let $\mathit{SYS}_0$ be the result of removing the first equation from $\textit{SYS}$ and adding $X_1$ to the free variables, $\mathcal{Y}$.
	If 
	for every environment $\rho$, $\trueset{F,\rho[\mathit{SYS}_0]} = \trueset{F',\rho[\mathit{SYS}_0]}$, then $\mathit{SYS'}$  is equivalent to the original system, $\mathit{SYS}$.
\end{lemma}

\begin{proof}
	Because of Lemma \ref{lem:systems:reorder}, it suffices to prove 
	the lemma when replacing the equation of the primary variable $X_1 = F_1$ by $X_1 = F_1'$.
	Then, 
	\begin{align*}
	\trueset{\mathit{SYS},\rho} &=\\ 
	&=
	 \bigcup \{ S_1 \mid S_1 \subseteq \trueset{ F_1,\overline{\rho[X_1 \mapsto S_1]}^{\mathit{SYS}_0 }} \}
	 \\
	 &=
	 \bigcup \{ S_1 \mid S_1 \subseteq \trueset{ F'_1,\overline{\rho[X_1 \mapsto S_1]}^{\mathit{SYS}_0} } \}
	\\ &= 
	\trueset{\mathit{SYS'},\rho}.
	\end{align*}
\qed\end{proof}

We now extend the notion of standard and deterministic forms to systems of equations. 

\begin{definition}
\label{def:std_form}
    Let $\mathit{SYS}$ be a system of equations that is equivalent to some formula $\varphi\in\shml$. We say that an equation, $X_i=F_i$ is in standard form if either $F_i = \texttt{ff}$, or
    \[
        F_i = \bigwedge_{j\in K_i}[\alpha_j]X_j\land\bigwedge_{j\in S_i}Y_j
    \]
    for some finite set of indices, $K_i$ and $S_i$.
    We say that $\mathit{SYS}$ is in standard form if every equation in $\mathit{SYS}$ is in standard form.
\end{definition}

\begin{lemma}
    \label{lemma:formula_to_std_sys}
    For each formula $\varphi\in\shml$, there exists a system of equations that is equivalent to $\varphi$ and is in standard form.
\end{lemma}
    \begin{proof}
        We use structural induction to show how we can construct a system of equations from a formula $\varphi$ that is in standard form. As shown in Lemma \ref{lemma:top_lvl_var}, given a formula $\vartheta\in\shml$ we can define an equivalent formula $\vartheta'$ where each free and unguarded variable is at the top level. We can therefore assume that for each fixed point $\max X.\psi$ that occurs as a subformula in $\varphi$, each free and unguarded variable in $\psi$ is at the top level of $\psi$ and using the equivalence $\max X.(\vartheta\land X) \equiv \max X.\vartheta$, we can also assume that $X$ does not appear at the top level of $\psi$.
        Furthermore, we assume that if $\varphi_1 \wedge \psi_2$ appears as a subformula of $\varphi$, then there is no variable which is free in $\varphi_1$ and bound in $\varphi_2$ (or vice-versa).

        We now go through the base cases and each of the top level operators that can occur in $\varphi$.

        \begin{description}
        \item[$\varphi=\texttt{tt}$:]
            We define a system of equations $\mathit{SYS}=(\{X=\texttt{tt} \land \texttt{tt}\}, X, \emptyset)$. Since $\bigwedge_{j\in\emptyset}\vartheta_j\equiv\texttt{tt}$, $\mathit{SYS}$ is in standard form and is equivalent to $\varphi$.

        \item[$\varphi=\texttt{ff}$:]
            We define a system of equations $\mathit{SYS}=(\{X=\texttt{ff}\}, X, \emptyset)$. $\mathit{SYS}$ is in standard form and is equivalent to $\varphi$.

        \item[$\varphi = Y$:]
            We define a system of equations $\mathit{SYS}=(\{X=Y \wedge \texttt{tt}\}, X, \{Y\})$. $\mathit{SYS}$ is in standard form and is equivalent to $\varphi$.

        \item[{$\varphi = [\alpha]\psi$}:]
            By the induction hypothesis, there exists a system of equations, $\mathit{SYS}=(Eq, X_1, \mathcal Y)$ that is equivalent to $\psi$ and is in standard form.
            
            We define a new system of equations
            \[
                \mathit{SYS}' = (Eq\cup\{X_0=[\alpha]X_1 \wedge \texttt{tt}\}, X_0, \mathcal Y)
            \]
            Each equation from $\mathit{SYS}'$ is in standard form and so $\mathit{SYS}'$ is in standard form. Since $\mathit{SYS}$ is equivalent to $\psi$ and $X_0$ does not appear in $\mathit{SYS}$ (so the first equation does not affect the fixed-point calculations), 
            \begin{align*}
            \llbracket \mathit{SYS}', \rho \rrbracket &= \\
            &=
            \trueset{[\alpha]X_1,\rho[X_1 \mapsto \trueset{\mathit{SYS},\rho}]}
            \\
            &=\{p \mid p \xRightarrow{\alpha} q \text{ implies } q \in \llbracket \mathit{SYS}, \rho \rrbracket \} \\
            &= \{p \mid p \xRightarrow{\alpha} q \text{ implies } q \in \llbracket \psi, \rho \rrbracket \} \\
            &= \llbracket [\alpha] \psi , \rho \rrbracket , 
            \end{align*}
			so $\mathit{SYS}'$
            is equivalent to $[\alpha]\psi$, which  is  $\varphi$.

        \item[{$\varphi = \psi_1\land\psi_2$}:]
            By the induction hypothesis, we know that there exist systems of equations $\mathit{SYS}_1=(Eq_1, X_1, \mathcal Y_1)$ and $\mathit{SYS}_2=(Eq_2, Z_1, \mathcal Y_2)$ that are equivalent to $\psi_1$ and $\psi_2$ respectively and are in standard form. Let $X_1=F_1$ be the principal equation from $\mathit{SYS}_1$ and let $Z_1=G_1$ be the principal equation from $\mathit{SYS}_2$.

            We define a new system of equations
            \[
                \mathit{SYS} = (Eq, X_0, \mathcal Y_1\cup\mathcal Y_2),
            \]
            where
            \[
                Eq = Eq_1\cup Eq_2\cup \{X_0=F_1\land G_1\}.
            \]
            Again, $X_0$ does not appear in $\mathit{SYS}_1$ or $\mathit{SYS}_2$, so it does not play a part in the fixed-point calculation; furthermore, for $X_i = F_i$ an equation in $Eq_1$, $X_i \notin \mathcal{Y}_2$, i.e. $X_i$ cannot be a free variable (or appear at all) in $\mathit{SYS}_2$ and vice-versa. Therefore, for $i = 1,2$ and $X_j = F_j$ an equation in $Eq_1 \cup Eq_2$, $\trueset{\mathit{SYS}_i,\rho[X_j \mapsto S_j]} = \trueset{\mathit{SYS}_i,\rho}$ (that is, $\rho(X_j)$ does not affect the computation of $\trueset{\mathit{SYS}_i,\rho}$) and therefore $\rho[\sys_1][\sys_2] = \rho[\sys_2][\sys_1]$ and for $i = 1,2$ and $j = 3-i$,
            \begin{equation}
            \trueset{ \sys_i,\overline{\rho[\mathit{SYS_i}]}^\mathit{SYS_{j}} } = 
            \trueset{ \sys_i,\overline{\rho[\mathit{SYS_{j}}]}^\mathit{SYS_{i}} }
            =
            \trueset{ \sys_i,\rho[\mathit{SYS_i}] }
            .
            \label{eq:and_eq_systems_independent}
            \end{equation}
            Finally,
            \begin{align*}
            \trueset{\mathit{SYS},\rho} &= \\
            &=
			\bigcup \left\{ S_0 \mid S_0 \subseteq \trueset{ F_1 \wedge G_1,\overline{\overline{\rho[X_0 \mapsto S_0]}^\mathit{SYS_1}}^\mathit{SYS_2} }   \right\}         
            \\
            &=
            \trueset{ F_1 \wedge G_1,\overline{\rho[\mathit{SYS_1}]}^\mathit{SYS_2} }     
            \ \ \quad  {(X_0 \text{ does not appear anywhere})}
            \\
            &=
            \trueset{ F_1,\overline{\rho[\mathit{SYS_1}]}^\mathit{SYS_2} }     
            \cap 
            \trueset{ G_1,\overline{\rho[\mathit{SYS_1}]}^\mathit{SYS_2} } 
            \\
            &=
            \trueset{ \mathit{SYS}_1,\rho[\mathit{SYS_2}] }     
            \cap 
            \trueset{ \mathit{SYS}_2,\rho[\mathit{SYS_1}] } 
            \ \ \quad 
            {(\text{by \eqref{eq:and_eq_systems_independent}})}
            \\
            &=
            \trueset{\mathit{SYS}_1,\rho} \cap \trueset{\mathit{SYS}_2,\rho}.
            \end{align*}
            Since $\mathit{SYS}_1$ is equivalent to $\psi_1$ and $\mathit{SYS}_2$ is equivalent to $\psi_2$, $\mathit{SYS}$ is equivalent to $\varphi=\psi_1\land\psi_2$.

            Both $X_1=F_1$ and $Z_1=G_1$ are in standard form and so we can write them as
            \begin{align*}
                F_1 &= \bigwedge_{j\in K_1}[\alpha_j]X_j\land\bigwedge_{j\in S_1}Y_j \\
                G_1 &= \bigwedge_{j\in K_1'}[\alpha_j]Z_j\land\bigwedge_{j\in S_1'}Y_j 
            \end{align*}

            This allows us to rewrite the equation for $X_0$ as follows:
            \begin{align*}
                X_0 &= F_1\land G_1 \\
                    &= \bigwedge_{j\in K_1}[\alpha_j]X_j\land\bigwedge_{j\in S_1}Y_j \land
                       \bigwedge_{j\in K_1'}[\alpha_j]Z_j\land\bigwedge_{j\in S_1'}Y_j \\
                    &\equiv \left(\bigwedge_{j\in K_1'}[\alpha_j]X_j\land\bigwedge_{j\in K_1'}[\alpha_j]Z_j\right)
                       \land\bigwedge_{j\in S_1\cup S_1'}Y_j 
            \end{align*}
            Now $\mathit{SYS}$ is in standard form and is equivalent to $\varphi$.

        \item[{$\varphi = \max Y.\psi$}:]
            By the induction hypothesis, there exists a system of equations $\mathit{SYS}=(Eq, X_1, \mathcal Y)$ that is equivalent to $\psi$ and is in standard form.

                If $\psi$ does not contain $Y$, then $\varphi\equiv\psi$ which means $\varphi$ is equivalent to $\mathit{SYS}$ and we are done.

            If $\psi$ does contain $Y$ then we define a new system of equation:
            \[
                \mathit{SYS}' = (Eq\cup \{Y=F_1\}, Y, \mathcal Y\setminus\{Y\})
            \]
            where $X_1 = F_1$ appears in $\mathit{SYS}$. 
            
            Let $\rho$ be an environment. Then, 

            \begin{align*}
            \trueset{\varphi,\rho} &= \\
            &=\bigcup \{S_0 \mid S_0 \subseteq \trueset{\psi,\rho[X_0 \mapsto S_0]} \}\\
            &=
            \bigcup \{S_0 \mid S_0 \subseteq \trueset{\mathit{SYS},\rho[X_0 \mapsto S_0]} \}
            \\
            &=
            \bigcup \{ S_0 \mid S_0 \subseteq \trueset{ F_1,\overline{\rho[X_0 \mapsto S_0]}^\sys  } \}
            \\
            &=
            \trueset{\mathit{SYS}',\rho}.
            \end{align*}

            By our assumption that for each maximum fixed point $\max X.\psi$ in $\varphi$, $X$ does not appear unguarded in $\psi$, we know that $Y$ does not appear unguarded in $F_1$.
            
            However in general we cannot guarantee that $Y$ does not appear unguarded in the equations from $\mathit{SYS}$, since $Y\in\mathcal Y$. To overcome this, we replace each unguarded occurrence of $Y$ with its corresponding formula $F_1$. Let $X_i=F_i$ be a formula that contains an unguarded occurrence of $Y$. Since $X_i=F_i$ is in standard form in $\mathit{SYS}$, we have
            \begin{align*}
                X_i &= \bigwedge_{j\in K_i}[\alpha_j]X_j\land\bigwedge_{j\in S_i}Y_j \\
                    &\equiv \bigwedge_{j\in K_i}[\alpha_j]X_j\land\bigwedge_{j\in S_i\setminus\{t\}}Y_j~~~\land Y_t 
            \end{align*}
            where $Y=Y_t$. We now change the equation for $X_i$ by replacing the unguarded occurrence of $Y_t$ with $F_1$ (by Lemma \ref{lem:replacement_lemma}):
            \begin{align*}
                X_i &= \bigwedge_{j\in K_i}[\alpha_j]X_j\land\bigwedge_{j\in S_i\setminus\{t\}}Y_j~~~\land F_1 \\
                    &= \bigwedge_{j\in K_i}[\alpha_j]X_j\land\bigwedge_{j\in S_i\setminus\{t\}}Y_j\land 
                       \bigwedge_{j\in K_1}[\alpha_j]X_j\land\bigwedge_{j\in S_1}Y_j \\
                    &\equiv \bigwedge_{j\in K_i\cup K_1}[\alpha_j]X_j\land\bigwedge_{j\in (S_i\setminus\{j\})\cup S_1}Y_j 
            \end{align*}
            (for simplicity, assume $K_1$, $K_i$ are disjoint) and the $i$'th equation is in standard form.

            Since $X_1 = F_1$ is in standard form in $\mathit{SYS}$, $Y = F_1$ is in standard form in $\mathit{SYS}'$. For all other equations from $\mathit{SYS}$, we can define equivalent equations that are in standard form in $\mathit{SYS}'$ by replacing every unguarded occurrence of $Y$ with $F_1$. All equations in $\mathit{SYS}'$ are now in standard form and since $\mathit{SYS}'$ is equivalent to $\varphi$, this case holds.
			\qed
        \end{description}
   \end{proof}

\begin{examp}
	The system of equations which which results from the construction of Lemma \ref{lemma:formula_to_std_sys} for formula $\varphi_e = \max\ X.[\alpha]([\alpha]\false \wedge X)$ is $\sys'_e$, which has no free variables and includes the equations 
	\begin{align*}
	X &= [\alpha] X_1, \\
	X_1 &= [\alpha] X_2 \wedge [\alpha] X_1,\\
	X_2 &= \false ,
	\end{align*} 
	where $X$ is the principal variable of $\sys'_e$.
\end{examp}

\begin{definition}
    Let $\mathit{SYS}=(Eq, X_1, \mathcal Y)$ be a system of  equations  equivalent to a formula in $\shml$. We say that an equation $X=\bigwedge F$ in $Eq$ is in deterministic form iff:
    for each pair of expressions $F_1, F_2 \in F \setminus \mathcal{Y}$,
    it must be the case that $ F_1 = [\alpha_1]X_i$ and $F_2 = [\alpha_2]X_j$ for some $\alpha_1,\alpha_2 \in \mycap{Act}$ and if $X_i \neq X_j$, then $\alpha_1 \neq \alpha_2$.
    We say that $\mathit{SYS}$ is in deterministic form if every equation in $Eq$ is in deterministic form.
\end{definition}

\begin{lemma}
    \label{lemma:det_system}
    For each $\shml$ system of equations in standard form, there exists an equivalent system of equations that is in deterministic form.
\end{lemma}
    \begin{proof}
        Let $\mathit{SYS}=(Eq, X_1, \mathcal Y)$ be a system of $n$ equations in standard form that is equivalent to a formula $\varphi\in\shml$ and 
        $$Eq = (X_1 = F_1,\ X_2=F_2,\ldots ,\ X_n = F_n) .$$
        
        We define some useful functions:
        \begin{align*}
            S(i)        &= \{\alpha_j~|~\text{$[\alpha_j]X_j$ is a sub formula in $F_i$}\}  \\
            D(i,\alpha) &= \{r~|~\text{$[\alpha]X_r$ is a sub formula in $F_i$}\}           \\
            E(i)        &= \{r~|~\text{$Y_r$ is unguarded in $F_i$}\}                       \\
        \end{align*}

        We also define these functions for subsets of indices $Q\subseteq\{1,2,\dots,n\}$:
        \begin{align*}
            S(Q)        &= \bigcup_{i\in Q}S(i)   ,      \\
            D(Q,\alpha) &= \bigcup_{i\in Q}D(i,\alpha) , \text{ and} \\
            E(Q)        &= \bigcup_{i\in Q}E(i)         .
        \end{align*}

        For each equation $X_i = F_i$ where $F_i\neq\texttt{ff}$, using these functions we can rewrite the equation as follows:
        \[
            X_i = \bigwedge_{\alpha\in S(i)}\left([\alpha]\bigwedge_{j\in D(i,\alpha)}X_j\right)\land\bigwedge_{j\in E(i)}Y_i 
        \]

        This equation may not be in deterministic form since it contains the conjunction of variables $\bigwedge_{j\in D(i,\alpha)}X_j$ and $D(i,\alpha)$ may not be a singleton.
        To fix this, we define a new variable $X_Q$ for each subset $Q\subseteq\{1,2,\dots,n\}$, such that $X_Q$ behaves like $\bigwedge_{j\in Q}X_j$; we identify $X_{\{i\}}$ with $X_i$.

         If for any $j\in Q$ we have $F_i = \texttt{ff}$ then $\bigwedge_{j\in Q}F_j \equiv \texttt{ff}$ and the equation for $X_Q$ is $X_Q=\texttt{ff}$. Otherwise, the  equation for $X_Q$ is:
        \[
        X_Q = \bigwedge_{\alpha\in S(Q)}\left([\alpha] \bigwedge_{i \in D(Q,\alpha)} X_i \right)\land \bigwedge_{j\in E(Q)}Y_j \ \ (=F_Q).
        \]
        
         We apply the following steps, each of which preserves the two conditions:
        \begin{enumerate}
        	\item the resulting system is equivalent to the original system and
        	\item for every equation $X_A = F'_A$, where $A \subseteq \{1,2,\ldots,n\}$, in the current system $\sys$ and environment $\rho$,  $\trueset{X_A,\rho[\sys]} = \trueset{F_A,\rho[\sys]}$. Note that it may be the case that $F'_A \neq F_A$, as we may have replaced $F_A$ in the original equation for $X_A$ by $F'_A$. This condition assures us that it doesn't matter, because $F_A$ and $F'_A$ are semantically  equivalent.
        \end{enumerate}
    Notice that condition 2 implies that $\trueset{X_A,\rho[\sys]} = \trueset{\bigwedge_{i \in A}X_i,\rho[\sys]}$.
        
        Consider an equation 
        \[
        X_Q = \bigwedge_{\alpha\in S(Q)}\left([\alpha]\bigwedge_{i\in D(Q,\alpha)} X_i\right)\land\bigwedge_{i\in E(Q)}Y_i ,
        \]
        which is already in the system and is not in deterministic form. If there is no equation for $X_{D(Q,\alpha)}$ in the system, then we introduce the equation for $X_{D(Q,\alpha)}$ as defined above (this gives an equivalent system, because $X_{D(Q,\alpha)}$ does not appear in any other equation if its own equation is not already in the system and condition 2 is preserved trivially).
        
        Let $S_1(Q) \cup S_2(Q) = S(Q)$ be such that for every $\alpha \in S_1(Q)$, $Q = D(Q,\alpha)$ and for every $\alpha \in S_2(Q)$, $Q \neq D(Q,\alpha)$. 
        Then,
		 let $\sys_0$ be the result of removing the equation for $X_Q$ from the system (and adding $X_Q$ to the free variables)   and $\sys'$ the result of replacing it by 
        \[
        X_Q = \bigwedge_{\alpha\in S(Q)}[\alpha]{X_{ D(Q,\alpha)}} \land\bigwedge_{j\in E(Q)}Y_i 
        \ \ (= F'_Q)
        .
        \]
        
        We claim that $\sys$ and $\sys'$ are equivalent  and after proving this claim we are done, because we can repeat these steps until all equations are in deterministic form and we are left with an equivalent deterministic system.
        To prove the claim, it is enough to prove that for every environment $\rho$, $\trueset{X_Q,\rho[\sys]} = \trueset{X_Q,\rho[\sys']}$ (by Lemma \ref{lem:systems:reorder}).
        Equivalently, we show that $A = B$, where
        \begin{align*}
        A &= \bigcup \left\{ S \mid S \subseteq \trueset{F_Q, \overline{\rho[X_Q \mapsto S]}^{\sys_0}} \right\}
        = \trueset{X_Q,\rho[\sys]}
        \\
        &\text{and}
        \\
        B &= \bigcup \left\{ S \mid S \subseteq \trueset{F'_Q, \overline{\rho[X_Q \mapsto S]}^{\sys_0}} \right\} = \trueset{X_Q,\rho[\sys']}.
        \end{align*}

        Thus, it suffices to prove that $A \subseteq B$ and $B \subseteq A$. Notice that
        \begin{align*}
        A &= \trueset{F_Q, \overline{\rho[X_Q \mapsto A]}^{\sys_0}} , \\
        B &= \trueset{F'_Q, \overline{\rho[X_Q \mapsto B]}^{\sys_0}} ,
        \end{align*}
        and that $\overline{\rho[X_Q \mapsto A]}^{\sys_0} = \rho[\sys]$ and 
        $\overline{\rho[X_Q \mapsto B]}^{\sys_0} = \rho[\sys']$. Therefore, it suffices to prove:
        \begin{align*}
        A & \subseteq 
        \trueset{F'_Q, \rho[\sys]} \\
        \text{and }\  B & \subseteq   \trueset{F_Q, \rho[\sys']} .
        \end{align*}
        For the first direction,
        \begin{align*}
        \trueset{F'_Q,\rho[\sys]} &= \\
        &= \trueset{\bigwedge_{\alpha\in S(Q)}\left([\alpha]X_{ D(Q,\alpha)}\right)\land\bigwedge_{j\in E(Q)}Y_i,\rho[\sys]} 
        \\
        &=
        \trueset{\bigwedge_{\alpha\in S(Q)}\left([\alpha] \bigwedge_{i \in D(Q,\alpha)} X_i \right)\land \bigwedge_{j\in E(Q)}Y_j,\rho[\sys]}\\
        &
        \text{(because of preserved condition 2)}
        \\
        &=
        \trueset{F_Q,\rho[\sys]} = A.
        \end{align*}
        
        On the other hand,
        \begin{align*}
                &        \trueset{F_Q,\rho[\sys']} = \\
        &= 
        \trueset{\bigwedge_{\alpha\in S(Q)}\left([\alpha] \bigwedge_{i \in D(Q,\alpha)} X_i \right)\land \bigwedge_{j\in E(Q)}Y_j,\rho[\sys']}
        \\
        &= 
        \trueset{\bigwedge_{\alpha\in S_1(Q)}\left([\alpha] \bigwedge_{i \in Q} X_i \right)\land \bigwedge_{\alpha\in S_2(Q)}\left([\alpha] \bigwedge_{i \in D(Q,\alpha)} X_i \right)\land \bigwedge_{j\in E(Q)}Y_j,\rho[\sys']}
        \\
        &= 
        \trueset{\bigwedge_{\alpha\in S_1(Q)}\left([\alpha] \bigwedge_{i \in Q} X_i \right)\land \bigwedge_{\alpha\in S_2(Q)}[\alpha] X_{ D(Q,\alpha)} \land \bigwedge_{j\in E(Q)}Y_j,\rho[\sys']}
		,
        \end{align*}
        because if $Q \neq  D(Q,\alpha)$, then $\trueset{X_{ D(Q,\alpha)},\rho[\sys_0]} = \trueset{ \bigwedge_{i \in D(Q,\alpha)} X_i ,\rho[\sys_0]}$, by the preserved condition 2.
        If $Q$ is a singleton, then we are done, because then $ \bigwedge_{i \in Q} X_i = X_Q$ and the last expression is just $B$.
        Therefore, we assume $Q$ is not a singleton; so, for all $i \in Q$, $Q \neq \{i\}$ and thus, $\trueset{X_i,\rho[\sys']} = \trueset{F_i,\rho[\sys']}$.
		For convenience, let 
		\begin{align*}
		C =& \trueset{\bigwedge_{\alpha\in S_2(Q)}\left([\alpha] \bigwedge_{i \in D(Q,\alpha)} X_i \right)\land \bigwedge_{j\in E(Q)}Y_j,\rho[\sys']}\\
		=& \trueset{\bigwedge_{\alpha\in S_2(Q)}[\alpha] X_{ D(Q,\alpha)}\land \bigwedge_{j\in E(Q)}Y_j,\rho[\sys']}
		. 
		\end{align*}
		
		Then, let $\sys'_Q$ be $\sys'$ after removing the equations for all $X_i$, $i \in Q$ and inserting all  $X_i$, where $i \in Q$ in the set of free variables, $\mathcal{Y}$.
		\begin{align*}
		\trueset{F_Q,\rho[\sys']} 
		&= 
		\trueset{\bigwedge_{i \in Q} F_i ,\rho[\sys']}  
		\ \ 
		\text{ (by definition)}
			\\
		&= 
		\trueset{\bigwedge_{i \in Q} X_i ,\rho[\sys']}  
		\ \
		= 
		\bigcap_{i \in Q}\trueset{ X_i ,\rho[\sys']}  
		\\
		&= 
		\bigcap \bigtimes_{i \in Q}\trueset{ X_i ,\rho[\sys']}  
		= 
		\bigcap \trueset{ \bigtimes_{i \in Q}X_i ,\rho[\sys']}  
		\\
		&= 
		\bigcap \bigcup \left\{ \mathbb{T}  \mid \mathbb{T} \subseteq  \trueset{ \bigtimes_{i \in Q}F_i ,\overline{\rho\left[\bigtimes_{i \in Q}X_i \mapsto \mathbb{T}\right]}^{\sys'_Q}} \right\}
		,
		\end{align*}
		because of Lemma \ref{lem:systems:simultaneous}.
		Similarly as to how we analyzed $F_Q$,
				\begin{align}
		B = \trueset{F'_Q,\rho[\sys']} = 
		\trueset{\bigwedge_{\alpha\in S_1(Q)}[\alpha] X_{Q} ,\rho[\sys']} \cap C
		.
		\label{eq:Banalysed}
		\end{align}
		So, it suffices to prove that for $k = |Q|$, 
		$$B^k \subseteq \trueset{ \bigtimes_{i \in Q}F_i ,\overline{\rho\left[\bigtimes_{i \in Q}X_i \mapsto B^k\right]}^{\sys'_Q}}.$$ 

		Let $p = (p_1,\ldots,p_k) \in \trueset{F'_Q,\rho[\sys']}^k = B^k$. 
		By \eqref{eq:Banalysed}, $p \in C^k$. Therefore, to prove that
		$$p \in  \trueset{ \bigtimes_{i \in Q}F_i ,\overline{\rho\left[\bigtimes_{i \in Q}X_i \mapsto B^k\right]}^{\sys'_Q}},$$ 
		it suffices to prove that for all $1 \leq i \leq k$, 
		\[
			p_i \in \trueset{ \bigwedge_{\alpha \in S(i) \cap S_1(Q)} [\alpha] \bigwedge_{j \in D(i,\alpha)  }X_j ,\overline{\rho\left[\bigtimes_{i \in Q}X_i \mapsto B^k\right]}^{\sys'_Q}},
		\]
		or equivalently that for every $\alpha \in S(i) \cap S_1(Q)$ and  $j \in D(i,\alpha) \subseteq D(Q,\alpha)$,
		\begin{equation}
		p_i \in \trueset{  [\alpha] X_j ,\overline{\rho\left[\bigtimes_{i \in Q}X_i \mapsto B^k\right]}^{\sys'_Q}}.
		\label{eq:pi_in_smth}
		\end{equation}
		For any $\alpha \in S(i) \cap S_1(Q)$, $j \in D(i,\alpha) \subseteq D(Q,\alpha) = B$, and $p_i \xRightarrow{\alpha} q_i$, because $p_i \in B$ and because of \eqref{eq:Banalysed}, 
		$
		p_i \in \trueset{
			[\alpha] X_{Q} ,\rho[\sys']},
		$ so 
		$q_i \in \trueset{X_Q,\rho[\sys']} = B$. Therefore,
		$$q_i \in  \trueset{X_j,\overline{\rho\left[\bigtimes_{i \in Q}X_i \mapsto B^k\right]}^{\sys'_Q}} = B,$$
		which gives us \eqref{eq:pi_in_smth} and the proof is complete.
	\qed\end{proof}

\begin{examp}
	The deterministic form of $\sys'_e$ is 
		is $\sys''_e$, which has no free variables and includes the equations 
		\begin{align*}
		X &= [\alpha] X_1, \\
		X_1 &= [\alpha] X_{12} ,\\
		X_2 &= \false ,\\
		X_{12} &= \false,
		\end{align*} 
		where $X$ is the principal variable of $\sys'_e$.
\end{examp}

\begin{lemma}
    \label{lemma:det_formula}
    Let $\mathit{SYS}=(Eq, X_1, \mathcal Y)$ be a $\shml$ system of equations in deterministic form. There exists a formula $\varphi\in\shml$ that is in deterministic form and is equivalent to $\mathit{SYS}$.
\end{lemma}
    \begin{proof}%
        We use proof by induction on the number of equations in $\mathit{SYS}$.

            For the base case, we assume that $\mathit{SYS}$ contains a single equation, $X_1=F_1$. Since $X_1=F_1$ is the principal equation in $\mathit{SYS}$, $\mathit{SYS}$ is equivalent to the formula $\max X_1.F_1$, which is in deterministic form. 

                        Now assume that $\mathit{SYS}$ contains $n>1$ equations. 
                        Let $\mathit{SYS}'$ be the result of removing the first equation from $\mathit{SYS}$ and adding $X_1$ to $\mathcal{Y}$ if $X_1$ appears in the remaining equations of $\mathit{SYS}'$.
                        Then, 
                        \[
                        \trueset{SYS,\rho} = \bigcup \left\{
						                        S_1 \mid
						                        S_1 \subseteq 
						                        \trueset{F_1,
						                        	\overline{\rho[X_1 \mapsto S_1]}^{\sys'}
						                        }
                        \right\};
                        \]
                        by the inductive hypothesis, there are formulae $\varphi_2,\ldots,\varphi_n$ in deterministic form, such that
                        \begin{align*}
                        \trueset{SYS,\rho} &= \\
                        &= \bigcup \{
                        S_1 \mid
                        S_1 \subseteq 
                        \trueset{F_1[\varphi_2/X_2,\ldots,\varphi_n/X_n],
                        	\rho[X_1 \mapsto S_1]} \\
                        &= 
                        \trueset{\max X_1.F_1[\varphi_2/X_2,\ldots,\varphi_n/X_n],\rho}
                        \}
                        \end{align*}
                         and $\max X_1.F_1[\varphi_2/X_2,\ldots,\varphi_n/X_n]$ is in deterministic form.
   \qed\end{proof}

Finally, we are ready to prove the main theorem in this section.

\begin{theorem}
    For each formula $\varphi\in\shml$ there exists a formula $\psi\in\shml$ that is equivalent to $\varphi$ and is in deterministic form.
\end{theorem}
    \begin{proof}
        Follows from Lemmata \ref{lemma:formula_to_std_sys}, \ref{lemma:det_system} and \ref{lemma:det_formula}.
   \qed\end{proof}

\begin{examp}
	If we follow through all the constructions, starting from $\varphi_e$, using Lemma \ref{lemma:formula_to_std_sys} we construct $\sys'_e$, which is a system of equations in standard form; from $\sys'_e$ we construct using Lemma \ref{lemma:det_system} we construct system $\sys''_e$, which is in deterministic form; finally, from $\sys''_e$ we can construct, using Lemma \ref{lemma:det_formula}, the deterministic form of $\varphi_e$, which happens to be $\varphi'_e = [\alpha][\alpha]\false$.
\end{examp}

This section's conclusion is that to monitor an $\shml$ formula, it is enough to consider deterministic monitors. 
On the other hand, if we are concerned with the computational cost of constructing a monitor (and we are), it is natural to ask what the cost of determinizing a monitor is. The following Section \ref{sec:bounds} is devoted to answering this and related questions.

\section{Bounds for Determinizing Monitors}
\label{sec:bounds}

The purpose of this section is to establish bounds on the sizes of monitors, both with respect to finite automata and when comparing deterministic to nondeterministic monitors.
The constructions of Section \ref{sec:rewriting} are not easy to extract bounds from. Therefore, we examine a different approach in this section. We remind the reader that in this section, we consider monitors which use only the \yes\ verdict --- although treating monitors which use only the \no\ verdict instead is completely analogous and we can also treat monitors which use both verdicts as described in Section \ref{sec:quick_fix}.

\subsection{Semantic Transformations}

We slightly alter the behavior of monitors 
to simplify this section's arguments. Specifically,
we provide three different sets of rules to define the behavior of the monitors, but we prove these are equivalent with respect to the traces which can reach a \yes\ value, which is what we need for this section.
Consider a single monitor, which appears at the beginning of the derivation under consideration --- that is, all other monitors are submonitors of this. We assume without loss of generality that each variable $x$ appears in the scope of a unique monitor of the form $\rec\ x.m$, which we call $p_x$; $m_x$ is the monitor such that $p_x = \rec\ x.m_x$. 
The monitors may behave according 
a system of rules.
System O is the old system of rules, as given in Table \ref{tab:Orules} of Section \ref{sec:def}; system N 
is
given by replacing rule \mycap{mRec} by the rules given by Table \ref{tab:Nrules}.

\begin{table}[h]
	\begin{align*}
	\mycap{mRecF}\frac{}{\texttt{rec} x.m_x\xrightarrow{\tau}m_x}   &&&&
	\mycap{mRecB}\frac{}{x\xrightarrow{\tau}p_x}   
	\end{align*}
	\caption{System N is the result of replacing rule \mycap{mRec} by rules \mycap{mRecF} and \mycap{mRecB}.}
	\label{tab:Nrules}
\end{table}

Derivations $\longrightarrow$ and $\Rightarrow$ are defined as before, but the resulting relations are called $\longrightarrow_O$ and $\Rightarrow_O$, and $\longrightarrow_N$ and $\Rightarrow_N$, respectively for systems O and N.
This subsection's main result, as given by Corollary \ref{cor:all_systems_eq}, is that systems O and N are equivalent. That is, for any monitor $p$, trace $t$, and value $v$, 
$$p \xRightarrow{t}_O v \ \text{ if and only if }\ p \xRightarrow{t}_N v.$$

To prove the equivalence of the two systems, we introduce an intermediate system, which we call system M. This is the result of replacing rule \mycap{mRec} of system O by the rules given by Table \ref{tab:Mrules}.

\begin{table}[h]
	\begin{align*}
	\mycap{mRecF}\frac{}{\texttt{rec} x.m_x\xrightarrow{\tau}m_x}   &&&&
	&\mycap{mRecP}\frac{}{x\xrightarrow{\tau}m_x}  
	\end{align*}
	\caption{System M is the result of replacing rule \mycap{mRec} by rules \mycap{mRecF} and \mycap{mRecP}.}
	\label{tab:Mrules}
\end{table}

For system M as well, derivations $\longrightarrow$ and $\Rightarrow$ are defined as before, but the resulting relations are called $\longrightarrow_M$ and $\Rightarrow_M$ for system M, to distinguish them from $\longrightarrow_O$, $\Rightarrow_O$, $\longrightarrow_N$, and $\Rightarrow_N$. 
Notice that the syntax used in system O and systems M and N is necessarily different.
While systems M and N require unique $p_x$ and $m_x$ and no substitutions occur in a derivation, the substitutions which occur in rule \mycap{mRec} produce new monitors and can result for each variable $x$ in several monitors of the form $\rec\ x.m$.
For example, consider monitor $$p = \rec\ x. (\alpha.\yes + \beta.\rec\ y.(\alpha.x +\beta.y)).$$ 
In this case, $p_x = p$ and $p_y = \rec\ y.(\alpha.x + \beta.y)$, but by rule \mycap{mRec} of system O, \[p\ \xrightarrow{\tau} \
\alpha.\yes + \beta.\rec\ y.(\alpha.\rec\ x. (\alpha.\yes + \beta.\rec\ y.(\alpha.x +\beta.y)) +\beta.y)
\]\[ \xrightarrow{\beta} \
\rec\ y.(\alpha.\rec\ x. (\alpha.\yes + \beta.\rec\ y.(\alpha.x +\beta.y)) +\beta.y) = p_y',
\]
but $p_y \neq p_y'$ and they are both of the form $\rec\ y.m$, which means that they cannot both appear in a derivation of system N or M.
Thus, when comparing the two systems, we assume one specific initial monitor $p_0$, from which all derivations in both systems are assumed to originate. 
Monitor $p_0$ satisfies the syntactic conditions for systems M and N and this is enough, since all transitions that can occur in these systems only generate submonitors of $p_0$.

The reason for changing the rules in the operational semantics of monitors is that for our succinctness results we need to track when recursion is used to reach a previous state of the monitor (i.e. a previous monitor in the derivation) by tracking when rule \mycap{mRecB} is used and then claim that this move takes us back a few steps in the derivation. Thus, we will be using system N for the remainder of this section. However, before that we need to demonstrate that it is equivalent to system O, in that for every monitor $p$, trace $t$, and value $v$, $p \xRightarrow{t}_O v$ iff $p \xRightarrow{t}_N v$. We prove this claim in this subsection by demonstrating equivalence of both systems to system M.

\begin{lemma}\label{lem:act_in_three_systems}
	Given monitors $p, q$ and action $\alpha \in \mycap{Act}$,
	$p \xrightarrow{\alpha}_N q$ iff $p \xrightarrow{\alpha}_M p'$ iff $p \xrightarrow{\alpha}_O p'$.
\end{lemma}

\begin{proof}
	Notice that in all three systems, $p \xrightarrow{\alpha} p'$ can only be produced by a combination of rules \mycap{mAct}, \mycap{mSelL}, and \mycap{mSelR} (all variants of \mycap{mRec} can only produce $\tau$-transitions, which are preserved through \mycap{mSelL} and \mycap{mSelR}). Since all three rules are present in all three systems, when this transition is provable in one of the systems, it can be replicated in each of them.
\qed\end{proof}

\begin{lemma}\label{lem:tau_transitions_have_x}
	If $p\xrightarrow{\tau}_M q$, then there is some variable $x$, such that $p$ is a sum of $x$ or of $p_x$ and $q = m_x$.
\end{lemma}

\begin{proof}
	Notice that $\tau$-transitions can only be introduced by rules \mycap{mRecF} and \mycap{mRecP} and then propagated by rules \mycap{mSelL} and \mycap{mSelR}. then, it is not hard to verify that if $p\xrightarrow{\tau}_M q$ was produced by \mycap{mRecF} or \mycap{mRecP}, then $p,q$ satisfy the property asserted by the lemma; and that rules \mycap{mSelL} and \mycap{mSelR} preserve this property.
\qed\end{proof}

\begin{lemma}\label{lem:tau_transitions_have_x_for_N}
	If $p\xrightarrow{\tau}_N q$, then there is some variable $x$, such that either $p$ is a sum of $x$ and $q = p_x$, or $p$ is a sum of $p_x$ and $q = m_x$.
\end{lemma}

\begin{proof}
	Very similar to the proof of Lemma \ref{lem:tau_transitions_have_x}.
\qed\end{proof}

\begin{lemma}\label{lem:tau_transitions_have_x_for_O}
	If $p\xrightarrow{\tau}_O q$, then there is a monitor $r = \rec\ x.r'$, such that  $p$ is a sum of $r$ and $q = r'[r/x]$.
\end{lemma}

\begin{proof}
Again, 	very similar to the proof of Lemma \ref{lem:tau_transitions_have_x}.
\qed\end{proof}

We first prove the equivalence of the systems M and N.

\begin{lemma}
	For a monitor $p$, trace $t$, and value $v$, $p \xRightarrow{t}_N v$ iff $p \xRightarrow{t}_M v$.
\end{lemma}

\begin{proof}
	We first prove that if $p \xRightarrow{t}_M v$, then $p \xRightarrow{t}_N v$ by induction on the length of the derivation  $p \xRightarrow{t}_M v$.
\begin{description}
		\item[If $p = v$,] then we are done, as the trivial derivation exists for both systems.
	\item[If  $p \xrightarrow{\alpha}_M p' \xRightarrow{t'}_M v$,] then $p \xrightarrow{\alpha}_N p'$ by Lemma \ref{lem:act_in_three_systems}. By the inductive hypothesis, $p' \xRightarrow{t'}_N v$ and we are done.
	\item[If $p \xrightarrow{\tau}_M p' \xRightarrow{t}_M v$,] then 
	by Lemma \ref{lem:tau_transitions_have_x}, there is some variable $x$, such that $p$ is a sum of $x$ or of $p_x$ and $p' = m_x$. Then, if $p$ a sum of $p_x$, $p \xrightarrow{\tau}_N m_x$ and if $p$ a sum of $x$, then $p \xrightarrow{\tau}_N p_x \xrightarrow{\tau}_N m_x = p'$. By the inductive hypothesis, $p' \xRightarrow{t}_N v$ and we are done.
\end{description}
We now prove that if $p \xRightarrow{t}_N v$, then $p \xRightarrow{t}_M v$, again by induction on the length of the derivation  $p \xRightarrow{t}_N v$. The only different case from above is for 
	when $p \xrightarrow{\tau}_N p' \xRightarrow{t}_N v$.
	In this case, by Lemma \ref{lem:tau_transitions_have_x_for_N}, there is some variable $x$, such that 
	$p$ is a sum of $x$ and $p' = p_x$, or $p$ is a sum of $p_x$ and $p' = m_x$.
	Then, if $p$ a sum of $p_x$, $p \xrightarrow{\tau}_M m_x = p'$ and by the inductive hypothesis, $p' \xRightarrow{t}_M v$ and we are done.
	If $p$ a sum of $x$, then $p' = p_x \neq v$, so the derivation is of the form $p \xrightarrow{\tau}_N p' = p_x \xrightarrow{\tau}_N m_x \xRightarrow{t}_N v$ (as $p_x$ can only transition to $m_x$), thus
	$p \xrightarrow{\tau}_M m_x$. By the inductive hypothesis, $m_x \xRightarrow{t}_M v$ and we are done.
\qed\end{proof}

Now we prove the equivalence of systems O and M.

\begin{lemma}
	For a monitor $p$, trace $t$, and value $v$, $p \xRightarrow{t}_M v$ iff $p \xRightarrow{t}_O v$.
\end{lemma}

\begin{proof}
	As we mentioned before, the syntax used in systems O and M is necessarily different. While rule $\mycap{mRecP}$ requires a unique $m_x$ (and thus a unique $p_x$), the substitutions which occur in rule \mycap{mRec} often result in several monitors of the form $\rec\ x.m$. The central idea of this proof is that if the derivation we consider starts with a monitor whose submonitors satisfy the uniqueness conditions we have demanded in systems $M$ and $N$ for $\rec\ x.m$ and then all produced monitors of the form $\rec\ x.m$ are somehow equivalent.
	
	Thus, we assume an initial monitor $p_0$, such that all derivations considered are subderivations of a derivation initialized at $p_0$ and such that every $x$ is bound in a unique submonitor $\rec\ x.m$, which is $p_x$ as required in system M. We define $\equiv$ to be the smallest equivalence relation such that 
	for all variables $x$, $x \equiv p_x$ and that for all $p$ and $q\equiv x$, $p \equiv p[q/x]$.
	We claim the following:
	\begin{description}[style=sameline]
		\item[For every $p = \rec\ x.q$ 
		which appears 
		in a derivation in O from $p_0$, $p \equiv p_x$:] 
		
		since only substitution can introduce new monitors of the form $\rec\ x.q$ and initially, for each variable $x$ the only such monitor is $p_x$, we can use induction on the overall number of  substitutions in the derivation from $p_0$ to prove the claim. If no substitutions happened, the claim is trivial. 
		Say the derivation so far is of the form $p_0 \xRightarrow{w}_O \rec\ y.r \xrightarrow{\tau}_O r[\rec\ y.r/y]$, with $r[\rec\ y.r/y]$ being the latest substitution in the derivation, such that 
		the claim holds for all submonitors of monitors appearing in 
		$p_0 \xRightarrow{w}_O \rec\ y.r$.
		Let $p = \rec\ x.q$ be a submonitor of $r[\rec\ y.r/y]$. We know by the inductive hypothesis that $\rec\ y.r \equiv p_y \equiv y$; then, there is some $p' = \rec\ x.q'$ submonitor of $\rec\ y.r$ (thus, $p' \equiv p_x$) and $p = p'[\rec\ y.r/y]$, but by the definition of $\equiv$, $p \equiv p' \equiv p_x$.
		
		\item[For a value $v$ and monitor $p$, if $v \equiv p$, then $p = v$:]
		 
			notice that $\equiv_M$ such that for value $v$, $v \equiv_M p$ iff $v = p$ and for $p,p'$ which are not values, always $p \equiv_M p'$, satisfies the above conditions, therefore $\equiv\ \subseteq\ \equiv_M$.

		\item[If $p+q \equiv r$, then there are $p'+q' = r$, where $p \equiv p'$ and $q \equiv q'$:] 
		
			notice that $\equiv$ can be constructed as the union $\bigcup_{i\in \NN}\equiv_i$, where $\equiv_0$ includes all pairs $(p,p)$, $(x,p_x)$, $(p_x,x)$, while $\equiv_{i+1}$ has all pairs in $\equiv_i$ and all $(p,p[q/x])$, $(p[q/x],p)$ and $(p,r)$, $(r,p)$, where $q \equiv_i x$ and $p \equiv_i p' \equiv_i r$;
			also, notice that neither $x$ not $p_x$ is of the form $p+q$ and all steps of the construction preserve the claim.

		\item[If $\alpha.p \equiv r$, then there is some $\alpha.q = r$, where $p \equiv q$:]
		 
		same as above.

		\item[If $\rec\ x.p \equiv r$ or $x \equiv r$, then $x = r$ or there is some $\rec\ x.q = r$, where $p \equiv q$:] 
		
		as above.
		
		\item[If $p \equiv r$ and $p$ is a sum of $q$, then $r$ is a sum of some $q' \equiv q$:] 
		
		by induction on the construction of a sum and the claim for $+$.
		
	\end{description}

	By the second claim, it is enough to prove that if $p \equiv p'$ and $p \xrightarrow{\mu}_O q$, then there is some $q' \equiv q$, such that $p' \xrightarrow{\mu}_M q'$ and vice-versa, if $p \equiv p'$ and $p' \xrightarrow{\mu}_M q'$, then there is some $q \equiv q'$, such that $p \xrightarrow{\mu}_O q$ (notice that this makes $\equiv$ a bisimulation).
	
	\begin{description}
		\item[For the first direction,] 
		let $p \equiv p'$ and $p \xrightarrow{\mu}_O q$. 
		
		If $\mu = \tau$, then by Lemma \ref{lem:tau_transitions_have_x_for_O}, $p$ is a sum of $\rec\ x.r$  and $q = r[p/x]$ --- from the last claim above, we assume for simplicity that $p = \rec\ x.r$ and it does not affect the proof. 
	By the claims above, either $p' = x$ or $p' = \rec\ x.q'$, where $r \equiv q'$. 
	Since $r \equiv r[p/x] = q$, if $p' = \rec\ x.q'$, then $p' = \rec\ x.q' \xrightarrow{\tau}_M q' \equiv r \equiv q$ and we are done.
	Otherwise, $p' = x$ and by the first claim, $p \equiv p_x$, so $m_x \equiv r \equiv q$, therefore $p' = x \xrightarrow{\tau}_M m_x \equiv q$ and we are done.
	
	If $\mu = \alpha \in \mycap{Act}$, then $p$ is a sum of some $\alpha.q$, so by the claims above, $p'$ is the sum of some $\alpha.q'$, where $q' \equiv q$; thus, $p' \xrightarrow{\mu}_M q'$.
	
	\item[For the opposite direction,] 
	let $p \equiv p'$ and $p' \xrightarrow{\mu}_M q'$. 
	
	If $\mu \in \mycap{Act}$, then the argument is as above. If $\mu = \tau$, then by Lemma \ref{lem:tau_transitions_have_x}, $p'$ is a sum of $x$ or $p_x$ and $q' = m_x$. Therefore, by the claims above, either $p = x$, which cannot occur in system O, or $p = \rec\ x.q$, for some $q \equiv m_x$. Therefore, $p \xrightarrow{\tau}_O q[p/x] \equiv q \equiv q'$. 
	\qed
	\end{description}
\end{proof}

\begin{corollary}\label{cor:all_systems_eq}
	For a monitor $p$, trace $t$, and value $v$, $p \xRightarrow{t}_N v$ iff $p \xRightarrow{t}_M v$ iff $p \xRightarrow{t}_O v$.
\end{corollary}

By Corollary \ref{cor:all_systems_eq}, systems M, N, and O are equivalent, so we will be using whichever is more convenient in the proofs that follow. For the remaining section, we use system N and we call $\rightarrow_N$ and $\Rightarrow_N$ simply $\rightarrow$ and $\Rightarrow$, respectively.

When using these new systems, we need to alter the definition of determinism. Notice that in System O, it is possible to have a nondeterministic monitor, which has a deterministic submonitor. Specifically, all variables are deterministic. This is fine in System O, because variables do not derive anything on their own. In systems M and N, though, a variable $x$ can derive $p_x$, so it is not a good idea to judge that any variable is deterministic --- and thus judge the determinism of a monitor only from its structure. 
For example, if $p_x = \rec\ x.(\alpha.x + \alpha.\yes)$, $x \xrightarrow{\tau}_N p_x$ and $x$ is deterministic (by default), while $p_x$ is not.
Of course, so far all sums encountered in the course of an O-derivation can be already observed in the initial monitor, so the definition of determinism so far was enough (in our example above, simply $x \not \rightarrow$). The issue now are free variables in a monitor $p$, which can derive supermonitors of $p$. Therefore, we demand of the initial monitor $p_0$,  that $p_0$ is deterministic (using the same definition as before).
Since it is the determinism of $p_0$ itself we are interested in, we will not see much difference from using the usual definition.

\subsection{Size Bounds for Monitors}

We present upper and lower bounds on the size of monitors. We first compare monitors to finite automata and then we examine the efficiency of 
monitor determinization. Note that monitors can be considered a special case of nondeterministic finite automata (NFA) and this observation is made explicit. 

This section's results are the following. We first provide methods to transform monitors to automata and back. 
One of the consequences of these transformations is that we can use the classic subset construction for the determinization of NFAs and thus determinize monitors. One advantage of this method over the one given in the previous sections is that it makes it easier to extract upper bounds on the size of the constructed monitors. 
Another is that it can be applied to an equivalent NFA, which can be smaller than the given monitor, thus resulting in a smaller deterministic monitor.
Then, we
demonstrate that there is an infinite family of languages $(L_n)_{n \in \NN}$, such that for each $n$, $L_n$ is recognizable by an NFA of $n+1$ states, a DFA of $2^n$ states, a monitor of size $O(2^n)$, and a deterministic monitor of size $2^{2^{O(n)}}$. Furthermore, we cannot do better, as we demonstrate that any  monitor which accepts $L_n$ must  be of size $\Omega(2^n)$ and every deterministic monitor which accepts $L_n$ must be of size $2^{2^{\Omega(n)}}$.

\subsubsection{From Monitors to Finite Automata}

A monitor can be considered to be a finite automaton with its submonitors as states  and $\Rightarrow$ as its transition relation. Here we make this observation explicit.\footnote{This is also possible, because system N only transitions to submonitors of an initial monitor; otherwise we would need to consider all monitors reachable through transitions and, perhaps, it would not be as clear which ones these are.} For a monitor $p$, we define the automaton $A(p)$ to be $(Q,\mycap{Act},\delta,q_0,F)$, where
\begin{itemize}
	\item $Q$, the set of states, is the set of submonitors of $p$;
	\item \mycap{Act}, the set of actions, is also the alphabet of the automaton;
	\item $q' \in \delta(q,\alpha)$ iff $q \Rightarrow\xrightarrow{\alpha} q'$;
	\item $q_0$, the initial state is $p$;
	\item $F = \{\yes \}$, that is, $\yes$ is the only accepting state.
\end{itemize}

\begin{proposition}\label{prp:A(p)is_good}
	Let $p$ be a monitor and $t \in \mycap{Act}^*$ a trace. Then, $A(p)$ accepts $t$ iff $t \xRightarrow{t} \yes$.
\end{proposition}

\begin{proof}
	We actually prove that for every $q \in Q$, $A_q(p) = (Q,\mycap{Act},\delta,q,F)$ accepts $t$ iff $q \xRightarrow{t} \yes$ and we do this by induction on $t$.
	If $t = \epsilon$, then $A_q(p)$ accepts iff $q \in F$ iff $q = \yes$ iff $q \Rightarrow \yes$. If $t = \alpha t'$, then\\
	  $A_q(p)$ accepts $t$
	  \\
	  \emph{iff} 
	  there is some $q' \in \delta(q,\alpha)$ such that $A_{q'}(p)$ accepts $t'$
	  \\ 
	  \emph{iff}
	  there is some $q'$ s.t. $q \Rightarrow\xrightarrow{\alpha} q'$ and $q' \xRightarrow{t'} \yes$
		\\
	  \emph{iff} 
	  $q \xRightarrow{t} \yes$. 
\qed\end{proof}

Notice that $A(p)$ has at most $|p|$ states (because $Q$ only includes submonitors of $p$), but probably fewer, since two occurrences of the same monitor as submonitors of $p$ give the same state --- and we can cut that down a bit by removing submonitors which can only be reached through $\tau$-transitions. Furthermore, if $p$ is deterministic, then $A(p)$ is deterministic.

\begin{corollary}
	For every (deterministic) monitor $p$, there is an (resp. deterministic) automaton which accepts the same language and has at most $|p|$ states.
\end{corollary}

\begin{corollary}
	All languages recognized by monitors are regular.
\end{corollary}

\subsubsection{From Automata to Monitors}

We would also like to be able to transform a finite automaton to a monitor and thus recognize regular languages by monitors. This is not always possible, though, since there are simple regular languages not recognized by any monitor. 
Consider, for example, $(11)^*$, which includes all strings of ones of even length. If there was such a monitor for the language, since $\epsilon$ is in the language, the monitor can only be \yes, which accepts everything (so, this conclusion is also true for any  regular language of the form $\epsilon + L \neq \mycap{Act}^*$).

One of the properties which differentiates monitors from automata is the fact that verdicts are irrevocable for monitors. Therefore, if for a monitor $p$ and finite trace $t$, $p \xRightarrow{t} \yes$, then for every trace $t'$, it is also the case that $p \xRightarrow{tt'} \yes$ (because of rule \mycap{mVerd}, which yields that for every $t'$, $\yes \xRightarrow{t'} \yes$). So, if $L$ is a regular language on \mycap{Act} recognized by a monitor, then $L$ has the property that for every  $t, t' \in \mycap{Act}^*$, if $t \in L$, then $tt' \in L$. We call such languages irrevocable (we could also call them suffix-closed). Now, consider an automaton which recognizes an irrevocable language. Then, if $q$ is any (reachable) accepting state of the automaton, notice that if we can reach $q$ through a word $t$, then $t$ is in the language and so is every $t \alpha$; therefore, we can safely add an $\alpha$-transition from $q$ to an accepting state (for example, itself) if no such transition exists.
We call an automaton which can always transition from an accepting state  to an accepting state irrevocable. Note that in an irrevocable DFA, all transitions from accepting states go to accepting states.

\begin{corollary}
	A language is regular and irrevocable if and only if it is recognized by an irrevocable NFA.
\end{corollary}

\begin{corollary}
	A language is regular and irrevocable if and only if it is recognized by an irrevocable DFA.
\end{corollary}
\begin{proof}
	Simply notice that the usual subset construction on an irrevocable NFA gives an irrevocable DFA.
\qed\end{proof}

	Let $A = (Q,\mycap{Act},\delta,q_0,F)$ be an automaton and $n = |Q|$.
For $1 \leq k \leq n$, $q_1,q_2,\ldots,q_k \in Q$ and $\alpha_1,\alpha_2,\ldots,\alpha_{k-1} \in \mycap{Act}$, $P = q_1\alpha_1 q_2\alpha_2\cdots \alpha_{k-1}q_k$ is a path of length $k$ on the automaton if all $q_1,q_2,\ldots,q_k$ are distinct and for $0 < i < k$, $q_{i+1} \in \delta(q_i,\alpha_i)$. Given such path, $P|_{q_i} = q_1,q_2,\ldots,q_i$ and $q_P = q_k$. By $q \in P$ we (abuse notation and) mean that $q$ appears in $P$.

\begin{theorem}\label{thm:nfa_to_monitor}
	Given an irrevocable NFA 
	of $n$ states, there is a monitor 
	of size $2^{O(n \log n)}$ which accepts the same traces as the automaton.
\end{theorem}

\begin{proof}
	Let $A = (Q,\mycap{Act},\delta,q_0,F)$ be an irrevocable finite automaton and $n = |Q|$. We can assume that $F$ has a single accepting state (since all states in $F$ behave the same way), which we call $Y$. 
	For some $q \in Q$, $A_q = (Q,\mycap{Act},\delta,q,F)$, which is the same automaton, but the run starts from $q$.
	Given a set $S$ of monitors, $\sum S$ is some sum of all elements of $S$.

	For every path $P = q_1\alpha_1\cdots \alpha_{k-1}q_k$ of length $k \leq n$ on $Q$,
	 we construct a monitor $p(P)$ by recursion on $n-k$.\\
	If $q_P = Y$, then $p(P) = \yes$. Otherwise,
	$$p(P) =  \rec\ x_P.\left( 
\begin{array}{c}
	\sum \{\alpha.p(P\alpha q) \mid 
	\alpha \in \mycap{Act} \text{ and } 
	q \in \delta(q_P,a) \text{ and } q \notin P \}\\
	+ \\
	\sum \{\alpha.x_{P|_q} \mid 
	\alpha \in \mycap{Act} \text{ and } 
	q \in \delta(q_P,a) \text{ and } q \in P \}
\end{array}
	\right) 
	.$$
	This was a recursive definition, because if $n = k$, all transitions from $q_P$ lead back to a state in $P$.
	Let $p = p(q_0)$. 
	Notice that $p$ satisfies our assumptions that all variables $x$ (i.e. $x_P$) are bound by a unique $p_x$ (i.e. $p(P)$). 
	Furthermore, following the definition above, for each path $P$ originating at $q_0$ we have generated at most $2|\mycap{Act}|\cdot n$ new submonitors of $p(P)$ --- this includes all the generated sums of submonitors of the forms $\alpha.p(P\alpha q)$ and $\alpha.x_{P_q}$ and $p(P)$ itself.
	Specifically, if the number of transitions from each state is at most $\Delta$ ($\Delta$ is at most $n |\mycap{Act}|$), then 
	for each path $P$ originating at $q_0$ we have generated at most $2\Delta$ new submonitors of $p(P)$.
	If the number of paths originating at $q_0$ is $\Pi$, then 
	$$|p| \leq 2 \Delta  \cdot \Pi.$$
	Let $P = q_1\alpha_1 q_2\alpha_2\cdots \alpha_{k-1}q_k$ be a path on $Q$. Notice that $q_1 q_2 \cdots q_k$ is a permutation of length $k$ of elements of $Q$ and $\alpha_1\alpha_2\cdots\alpha_{k-1}$ is a $(k-1)$-tuple of elements form $\mycap{Act}$.
	Therefore, the number of paths originating at $q_0$ is 
	at most $$\Pi \ \leq \ \sum_{k=1}^{n-1} |\mycap{Act}|^k \cdot (k!) \leq (n-1)!\cdot |\mycap{Act}|^{n-1} \cdot n \ = \ |\mycap{Act}|^{n-1} \cdot (n!)$$
	and thus,
	$$|p| \ \leq \ 2 n |\mycap{Act}|  \cdot 
	|\mycap{Act}|^{n-1} \cdot (n!)
	\ = \  2 n |\mycap{Act}|^{n} \cdot (n!)
	\ = \  2^{O(n \log n)}.$$
	To prove that $p$ accepts the same traces as $A$, we prove the following claim. 
	
	\emph{Claim:} for every such path $P$, $p(P) \xRightarrow{t} \yes$ if and only if $A_{q_P}$ accepts $t$. \\
	We prove the claim by induction on $t$. \\
	If $t = \epsilon$, then $p(P) \xRightarrow{t} \yes$ iff $p(P) = \yes$ iff $q_P = Y$. \\
	If $t = \alpha t'$ and $A_{q_P}$ accepts $t$, 
	then there is some $q \in \delta(q_P,\alpha)$, such that $A_q$ accepts $t'$.
	We consider the following cases: 
	\begin{itemize}
		\item[-] 
		if $q_P = Y$, then $p(P) = \yes \xRightarrow{t} \yes$;
		\item[-]
		if $q \in P$, then $p(P) \xrightarrow{\alpha} x_{P_q} \xrightarrow{\tau} p(P_q)$ and by the 
		inductive hypothesis, $p(P_q) \xRightarrow{t'}\yes$;
		\item[-] 
		otherwise, $p(P) \xrightarrow{\tau} \xrightarrow{\alpha} p(P\alpha q)$ 
		and by the 
		inductive hypothesis, $p(P\alpha q) \xRightarrow{t'}\yes$.
	\end{itemize}
	If $t = \alpha t'$ and $p(P) \xRightarrow{t} \yes$ and $q_P \neq Y$, 
	then there is some $q \in \delta(q_P,\alpha)$, such that
	either 
	$p(P) \xrightarrow{\tau} \xrightarrow{\alpha} p(P\alpha q) \xRightarrow{t'} \yes$ (when $q \notin P$), 
	or 
	$p(P) \xrightarrow{\tau} \xrightarrow{\alpha}  x_{P|_q} \xrightarrow{\tau} p(P|_q)  \xRightarrow{t'} \yes$ (when $q \in P$);
	in both cases, by the inductive hypothesis, $A_q$ accepts $t'$, so $A$ accepts $t$.
\qed\end{proof}

\begin{corollary}\label{cor:dfa_to_monitor} 
	Given an irrevocable DFA of $n$ states, there is a deterministic monitor of size at most  $2 n \cdot |\mycap{Act}|^{n} = 2^{O(n)}$ which accepts the same traces as the automaton.\footnote{Note that if $|\mycap{Act}| = 1$, then this Corollary claims a linear bound on the size of the deterministic monitor with respect to the number of states of the DFA. See Corollary \ref{cor:unary_is_easy} at the end of this section and the discussion right above it for more details.}
\end{corollary}

\begin{proof}
	Notice that the proof of Theorem \ref{thm:nfa_to_monitor} works for DFAs as well.
	We use the same construction.
	Unless $p = \yes$, every recursive operator (except the first one), variable, and value is prefixed by an action; furthermore, if $A$ is deterministic and $\alpha.p_1$, $\alpha.p_1$ appear as part of the same sum, then $p_1 = p_2$. 
	So, we have constructed a deterministic monitor.
	
	Since $A$ is deterministic, every path in $A$, $q_1\alpha_1 \cdots \alpha_{k-1}q_k$ is fully defined by the sequence of actions which appear in the path, $\alpha_1\alpha_2 \cdots \alpha_{k-1}$. Since $k \leq n$, the number of paths in $A$ is thus at most $$\sum_{k = 1}^{n} |\mycap{Act}|^{k - 1} < n \cdot |\mycap{Act}|^{n - 1}.$$
	As we mentioned in the proof of  Theorem \ref{thm:nfa_to_monitor}, if the number of paths originating at $q_0$ is $\Pi$ and the number of transitions from each state is at most $\Delta$, then 
	$$|p| \leq 2 \Delta  \cdot \Pi $$
	and therefore, 
	$$|p| \leq 2 n |\mycap{Act}|^n.
	$$
\qed\end{proof}

\begin{corollary}
	A language is regular and irrevocable if and only if it is recognized by a (deterministic) monitor.
\end{corollary}

\begin{corollary}\label{cor:determinization_upper}
	Let $p$ be a monitor for $\varphi$. Then, there is a deterministic monitor for $\varphi$ of size $2^{O\left(2^{|p|}\right)}$.
\end{corollary}

\begin{proof}
	Transform $p$ into an equivalent NFA of at most $|p|$ states, then to a DFA of $2^{|p|}$ states, and then, into an equivalent deterministic monitor of size $2^{O\left(2^{|p|}\right) }.$
\qed\end{proof}

Now, this is a horrible upper bound. 
Unfortunately, as the remainder of this section demonstrates, we cannot do much better.

\subsubsection{Lower Bound for (nondeterministic) Monitor Size}

The family of languages we consider is (initially) the following. For $n \geq 1$, let
$$L_n = \{ \alpha 1 \beta \in \{0,1\}^* \mid |\beta| = n-1 \}.$$ This is a well-known example of a regular language recognizable by an NFA of $n+1$ states, by a DFA of $2^n$ states, but by no DFA of fewer than $2^n$ states. 
As we have mentioned, monitors do not behave exactly the way automata do and can only recognize irrevocable languages. Therefore, we modify $L_n$ to mark the ending of a word with a special character, $e$, and make it irrevocable. Thus,
\[ M_n = \{\alpha e \beta \in \{0,1,e\}^* \mid \alpha \in L_n  \}.\footnote{Note that we can also allow for infinite traces without consequence.} \]

Note that an automaton (deterministic or not) accepting $L_n$ can easily be transformed into one (of the same kind) accepting $M_n$ by introducing two new states, $Y$ and $N$, where $Y$ is accepting and $N$ is not, so that all transitions from $Y$ go to $Y$ and from $N$ go to $N$ ($N$ is a junk state, thus unnecessary for NFAs); then we add an $e$-transition from all accepting states to $Y$ and from all other states to $N$. The reverse transformation is also possible: 
From an automaton accepting $M_n$, we can have a new one accepting $L_n$ by shedding all $e$-transitions and turning all states that can $e$-transition to an accepting state of the old automaton to accepting states.
The details are left to the reader.

Thus, there is an an NFA for $M_n$ with $n+2$ states and a DFA for $M_n$ with $2^n +2$ states, but no less.
We construct a (nondeterministic) monitor for $M_n$ of size $O(2^n)$. For every $\alpha \in \{0,1\}^*$, where $|\alpha| = k \leq n - 1$, we construct a monitor $p_\alpha$ by induction on $n-k$: if $k = n - 1$, 
$p_\alpha = e.\yes$; otherwise, $p_\alpha = 0.p_{\alpha 0} + 1.p_{\alpha 1}$. Let $p = \rec\ x.(0.x + 1.x + 1.p_\epsilon)$. Then, $p$ mimics the behavior of the NFA for $M_n$ and $|p| = 8 + |p_\epsilon | = O(2^n)$.

\begin{definition}
	We call a derivation $ p \xRightarrow{t} q$ simple, if rules \mycap{mRecB} and \mycap{mVerd} are not used in the proof of any transition of the derivation. 
	We say that a trace $t \in \mycap{Act}^*$ is simple for monitor $p$ if there is a simple  derivation $ p \xRightarrow{t} q$. We say that a set $G$ of simple traces for $p$ is simple for $p$.
\end{definition} 

\begin{lemma}
	Every subderivation of a simple derivation is simple.
\end{lemma}

\begin{corollary}
	If $t'\sqsubseteq t$ and $t$ is a simple trace for monitor $p$, then $t'$ is also a simple trace for $p$.
\end{corollary}

\begin{lemma}\label{lem:simple_traces}
	Let $p$ be a monitor and $G$ a (finite) simple set of traces for $p$. Then, $|p| \geq |G|$.
\end{lemma}

\begin{proof}
	By structural induction on $p$. 
	\begin{description}
		\item[If $p = v$ or $x$,] then $|G|$ is either empty or $\{\epsilon \}$ and $|p| \geq 1$. 
		\item[If $p = \alpha.q$,] then all non-trivial derivations that start from $p$, begin with $p \xrightarrow{\alpha} q$. Therefore, all traces in $G$, except perhaps for $\epsilon$, are of the form $\alpha t$. Let $G_\alpha = \{ t \in \mycap{Act}^* \mid \alpha t \in G \}$. Then, $|G| \leq |G_\alpha| + 1$, as there is a 1-1 and onto mapping from $G \setminus \{\epsilon \}$ to $G_\alpha$, namely $\alpha t \mapsto t$. By the inductive hypothesis, $|q| \geq |G_\alpha|$, so $|p| = |q| + 1 \geq |G_\alpha| +1 \geq |G|$.
		\item[If $p = q+r$,] then notice that all derivations that start from $p$ (including the trivial one, if you consider $p \Rightarrow p$ and $q\Rightarrow q$ to be the same) can also start from either $q$ or $r$. Therefore, $G = G_q \cup G_r$, where $G_q$ is simple for $q$ and $G_r$ is simple for $r$. Then, $|p| = |q| + |r|  \geq |G_q| + |G_r| \geq |G|$.
		\item[If $p = \rec\ x.q$,] then all non-trivial derivations that start from $p$, must begin with $p \xrightarrow{\tau} q$. Therefore, it is not hard to see that $G$ is simple for $q$ as well, so $|p| = |q| +1 > |G|$.
		\qed
	\end{description}
\end{proof}

\begin{corollary}\label{cor:height_to_simple}
	Let $t$ be a simple trace for $p$. Then, $h(p) \geq |t|$.
\end{corollary}
\begin{proof}
	Very similar to the proof of Lemma \ref{lem:simple_traces}.
\qed\end{proof}

\begin{corollary}
	Let $p$ be a monitor and $G$ a simple set of traces for $p$. Then, $G$ is finite.
\end{corollary}
\begin{proof}
	Notice that $|p| \in \NN$.
\qed\end{proof}

\begin{corollary}
	Let $p$ be a monitor and $t$ a simple trace for $p$. Then, $t$ is finite.
\end{corollary}
\begin{proof}
A direct consequence of Corollary \ref{cor:height_to_simple}.
\qed\end{proof}

\begin{lemma}\label{lem:find_px}
	In a derivation $p \xRightarrow{t} x$, such that $x$ is bound in $p$, a sum of $p_x$
	 must appear acting as $p_x$.
\end{lemma}

\begin{proof}
	By induction on the length of the derivation. If it is $0$, then $p = x$ and $x$ is not bound. Otherwise, notice that naturally, 
	$x$ is bound in $p$, but $x$ is not bound in $x$. Let $p \xRightarrow{t'} q \xrightarrow{\mu} q'$ be the longest initial subderivation such that 
	$x$ remains bound in $q$. Thus, it must 
	be the case that 
	$x$ is not bound in $q'$. The only rule which can have this effect is \mycap{mRecF},
	 possibly combined with \mycap{mSel}, 
	 so $q$ is a sum of $p_x$ acting as $p_x$.
\qed\end{proof}

The following lemma demonstrates that derivations which are not simple enjoy a property which resembles the classic Pumping Lemma for regular languages.

\begin{lemma}\label{lem:pump_not_simple}
	Let $p \xRightarrow{t} q$, such that $t$ is not simple for $p$. Then, there are $t = x u z$, such that $|u| >0$ and for every $i \in \NN$,  $p \xRightarrow{xu^i z} q$.
\end{lemma}

\begin{proof}
	Let $p \xRightarrow{t} q$ be a non-simple derivation $d$. We assume that there are no $s (\xrightarrow{\tau})^+ p_x$ parts in it, where $s$ is a sum of $p_x$ acting as $p_x$;
	otherwise remove them and the resulting derivation is non-simple, because $t$ is non-simple. Let $s \sqsubseteq t$ be the longest prefix of $t$, such that subderivation $p \xRightarrow{s} r$ is simple. Then, there is a $s\alpha \sqsubseteq t$, such that there is a subderivation of $d$, $p \xRightarrow{s} r \xrightarrow{\alpha} r'$ and in $r \xrightarrow{\alpha} r'$, either \mycap{mVerd} or \mycap{mRecB} is used and neither is used in $p \xRightarrow{s} r$. If \mycap{mVerd} is used, then $p \xRightarrow{s} v$ as part of $d$, so $q = v$ and $p \xRightarrow{su^i} v$, for $s u=t$ and $i \in \NN$. 
	If \mycap{mRecB} is used, then $p \xRightarrow{s} x \xrightarrow{\tau = \alpha} p_x$; by Lemma \ref{lem:find_px}, a sum of $p_x$ acting as $p_x$, say $s_x$ must appear in $p \xRightarrow{s} x$, so $s_x \xRightarrow{u} p_x$ is thus part of the derivation and $|u| >0$ by our assumption at the beginning of the proof. Then, for some $x,z$,  $t = x u z$ and for every $i \in \NN$,  $p \xRightarrow{xu^i z} q$.
\qed\end{proof}

\begin{corollary}\label{cor:pump_wrt_height}
	Let $p$ be a monitor and $t$ a trace, such that $|t| > h(p)$. Then, there are $t = x u z$, such that $|u| >0$ and for every $i \in \NN$,  $p \xRightarrow{xu^i z} q$.
\end{corollary}


\begin{proposition}\label{prp:nfa2mon_is_hard}
	Let $p$ be a monitor for $M_n$. Then, $|p| \geq 3 \cdot 2^{n-1}$.
\end{proposition}

\begin{proof}
	Because of Lemma \ref{lem:simple_traces}, it suffices to find a set $G$ of simple traces for $p$, such that $|G| \geq 3 \cdot 2^{n-1}$. We define 
	$$G = \{t \in \{0,1,e\}^* \mid t \sqsubseteq 1 s e, \text{ where }  s \in \{0,1\}^{n-1}\}.$$ 
	Then, $|G| \geq 3 \cdot 2^{n-1}$, so it suffices to demonstrate that all traces in $G$ are simple. 
	In turn, it suffices to demonstrate that for $s \in \{0,1\}^{n-1}$, $1 s e$ is simple. If it is not, then by Lemma \ref{lem:pump_not_simple}, since $p \xRightarrow{1se} \yes$, there is a (strictly) shorter trace $t$, such that $p \xRightarrow{t} \yes$, which contradicts our assumption that $p$ is a monitor for $M_n$.
\qed\end{proof}

We have thus demonstrated that for every $n \geq 1$, there is a monitor for $M_n$ of size $O(2^n)$ and furthermore, that there is no monitor for $M_n$ of size less than $3 \cdot 2^{n-1}$. So, to recognize languages $M_n$,  monitors of size exponential with respect to $n$ are required and thus we have a lower bound on the construction of a monitor from an NFA, which is close to the respective upper bound provided by Theorem \ref{thm:nfa_to_monitor}.

\subsubsection{Lower Bound for Deterministic Monitor Size}

We now consider deterministic monitors. We demonstrate (see Theorem \ref{thm:MnisbadforDmon}) that to recognize languages $M_n$ a deterministic monitors needs to be of size $2^{2^{\Omega(n)}}$. Therefore, a construction of a deterministic monitor from an equivalent NFA can result in a double-exponential blowup in the size of the monitor; constructing a deterministic monitor from an equivalent nondeterministic one can result in an exponential blowup in the size of the monitor. As Theorem \ref{thm:determinizing_monitors_is_VERY_hard} demonstrates, the situation is actually worse for the determinization of monitors, as there is a family $U_n$ of irrevocable regular languages, such that for $n\geq 1$, $U_n$ is recognized by a nondeterministic monitor of size $O(n)$, but by no deterministic one of size $2^{2^{o(\sqrt{n \log n})}}$. The proof of Theorem \ref{thm:determinizing_monitors_is_VERY_hard} relies on a result by Chrobak \cite{CHROBAK1986149} for unary languages (languages on only one symbol).

\begin{lemma}
	Let $p$ be a
	deterministic monitor. If $p \xRightarrow{t} q$, then $q$ is 
	deterministic.
\end{lemma}

\begin{proof}
	Both $p$ and $q$ are submonitors of the original monitor $p_0$. If $p$ is deterministic, then so is $p_0$, and so is $q$.
\qed\end{proof}

\begin{lemma}\label{lem:no_tau_from_det_sum}
	If $q+ q'$  is  deterministic, then $q+ q' \not \xrightarrow{\tau}$.
\end{lemma}

\begin{proof}
	The monitor $q+q'$ is a submonitor of the initial monitor, which is deterministic, so $q+q'$ must be of the form $\sum_{\alpha \in A} \alpha.p_\alpha$.
	We continue by induction on $q+q'$: if $q = \alpha.r$ and $q' = \beta.r'$, then by the production rules, only $q+ q' \xrightarrow{\alpha} r$ and $q+ q' \xrightarrow{\beta} r'$ are allowed; if one of $q$, $q'$ is also a sum, then by the derivation rules, if $q+ q' \xrightarrow{\tau}$, then also $q \xrightarrow{\tau}$ or $q' \xrightarrow{\tau}$, but by the inductive hypothesis, this is a contradiction.
\qed\end{proof}

\begin{corollary}\label{cor:all_det_tau}
	Only $\tau$-transitions of the form $x \xrightarrow{\tau} p_x$ and $\rec\ x.m \xrightarrow{\tau} m$ are allowed for  deterministic monitors.
\end{corollary}

\begin{lemma}
	Let $p \sim p'$ be  deterministic monitors. If $p \xRightarrow{t} q $ and $p' \xRightarrow{t} q'$, then $q \sim q'$.
\end{lemma}

\begin{proof}
	First, notice that $x \sim p_x$, since all derivations from $x$ are either trivial or must start with $x \xrightarrow{\tau} p_x$. For the same reason, $\rec\ x.m \sim m$. Therefore, by Corollary \ref{cor:all_det_tau}, if $r \xrightarrow{\tau} r'$, then $r \sim r'$.
	Now we can prove the lemma by induction on $|t|$.
	\begin{description}
		\item[If $t = \epsilon$,] then  $p \xRightarrow{} q $ and $p' \xRightarrow{} q'$, so (by the observation above) $q \sim p \sim p' \sim q'$.
		\item[If $t = \alpha t'$,] then we demonstrate that if $p \xRightarrow{\alpha} r $ and $p' \xRightarrow{\alpha} r'$, then $r \sim r'$ and, by induction, we are done. 
		In fact, by our observations, it is enough to prove that if $p \xrightarrow{\alpha} r $ and $p' \xrightarrow{\alpha} r'$, then $r \sim r'$. Thus, let $p \xrightarrow{\alpha} r $ and $p' \xrightarrow{\alpha} r'$; 
		then, $p$ is a sum  of $\alpha.r$ and of no other $\alpha.f$ (because $p$ is deterministic) and $p'$ is a sum  of $\alpha.r'$ and of no other $\alpha.f$.
		If $r \not \sim r'$, then there is a trace $s$ and value $v$, such that $r \xRightarrow{s} v$ and $r' \not \xRightarrow{s} v$ (or $r' \xRightarrow{s} v$ and $r \not \xRightarrow{s} v$, but this case is symmetric). Therefore, $p \xRightarrow{\alpha s} v$, but since all derivations from $p'$ on trace $\alpha s$ must start from $p' \xrightarrow{\alpha} r'$, if $p' \xRightarrow{\alpha s} v$, also $r' \xRightarrow{s} v$, so $p \not \sim p'$, a contradiction. Therefore, $r \sim r'$ and the proof is complete. 
		\qed
	\end{description}
\end{proof}

\begin{corollary}\label{cor:det_mon_is_det}
	If $p$ is deterministic, $p \xRightarrow{t} q $, and $p \xRightarrow{t} q'$, then $q \sim q'$.
\end{corollary}

\begin{corollary}\label{cor:determinism_gives_always_same_v}
	If $p$ is deterministic, $p \xRightarrow{t} q $, and $p \xRightarrow{t} v$, where $v$ is a value, then $q = v$.
\end{corollary}

\begin{lemma}\label{lem:no_sums}
	If $p$ is a deterministic monitor and $q$ is a submonitor of $p$, such that $q$ is a sum of some $p_x$, then $q = p_x$.
\end{lemma}

\begin{proof}
	By the definition of deterministic monitors, $r + p_x$ is not allowed.
\qed\end{proof}

\begin{corollary}\label{cor:find_px_determ}
	In a derivation $p \xRightarrow{t} x$, such that $x$ is bound in $p$ and $p$ is deterministic, $p_x$
	must appear.
\end{corollary}

\begin{proof}
	Combine Lemmata \ref{lem:find_px} and \ref{lem:no_sums}.
\qed\end{proof}

\begin{lemma}\label{lem:find_px_det}
	Let $p \xRightarrow{t} q$, such that $t$ is not simple for $p$ and $p$ is deterministic. Then, there are $t = x u z$ and monitor $r$, such that $|u| >0$ and  $p \xRightarrow{x} r \xRightarrow{u} r \xRightarrow{z} q$.
\end{lemma}

\begin{proof}
	Similar to the proof of Lemma \ref{lem:pump_not_simple}, but use Corollary \ref{cor:find_px_determ}, instead of Lemma \ref{lem:find_px}.
\qed\end{proof}

\begin{theorem}\label{thm:MnisbadforDmon}
	Let $p$ be a deterministic monitor for $M_n$. Then, $|p| =  2^{2^{\Omega(n)}}$.
\end{theorem}

\begin{proof}
	We construct a set of simple traces of size $\Omega((2^{n/3-1}-1)!)$ and by Lemma \ref{lem:simple_traces} this is enough to prove our claim.
	Let $k = \lfloor n /3 \rfloor - 1$ (actually, for simplicity we assume $n = 3 k +3$, as this does not affect our arguments); let 
	$$A = \{a \in \{0,1\}^k \mid \text{ at least one $1$ appears in } a \},$$ 
	let
	$$ B = \{ 1 a 1 0^k 1 a  \in \{0,1\}^{n} \mid a \in A \}, $$
	and let 
	$$ G = \{ x \in \{0,1\}^* \mid x \sqsubseteq g, \text{ where $g$ is a permutation of } B  \}; $$
	then, $|A| = |B| = 2^k-1$ and $|G| > |B|! = 2^{\Omega(|B| \log |B|)} = 2^{\Omega(k\cdot 2^{k})} = 2^{2^{\Omega(n)}}$. It remains to demonstrate that all permutations $g$ of $B$, are simple traces for $p$. 
	Let $g$ be such a permutation. Let $h = (2^k-1)!$ and $g = g_1 g_2\cdots g_h$, where $ g_1,g_2,\ldots, g_h \in B$ and for all such $g_i$, let $g_i = 1 b_i 1 0^k 1 b_i$.
	
	Notice that $g$ is designed so that in every subsequence of length $n$ of $g$, $0^k$ can appear at most once, specifically as the area of $0^k$ in the middle of a $g_i \in B$ as defined. If $0^k$ appears in an area of $k$ contiguous positions, then that area cannot include one of the separating 1's, so it must be an $a \in A$,  which cannot happen by the definition of $A$, or the area of $0^k$ in the middle of a $b \in B$ as defined. If there is another area of $0^k$ closer than $2 k+3$ positions away, again, it must be some $a \in A$, which cannot be the case. 
	Furthermore, and because of this observation, for every $a \in A$, $0^k 1 a$ and $a 1 0^k$ appear exactly once in $g$.
	
	For traces $x,y$, we say that $x \sim y$ if for every trace $z$, 
	$$x z \in M_n \ \text{ iff } \ y z \in M_n.$$
	For trace $x$, if $|x| \geq n$, we define $l(x)$ to be such that $|l(x)| = n$ and there is $x' l(x) = x$; if $|x| < n$, we define $l(x) = 0^{n - |x|}x$. 
	Notice that for traces $x,y$, $x \sim y$ iff $l(x) = l(y)$.
	Also, that if $|x| \geq n$, $l(x)$ must be in one of the three forms below: 
	\begin{enumerate}
		\item $l(x) = 0^{n_1}1 a_1 1 a_210^{n_2}$ for some $n_1 + n_2 = k$ and $a_1, a_2 \in A$ (in this case, $10^k 1 a_1 1 a_210^{n_2} = 10^{n_2}l(x)$ are the last $n+n_2$ positions of $x$), or 
		\item $l(x) = d_1 1 a_1 1 0^{k} 1 d_2$ for some $d_2 \sqsubseteq a_1 \in A$, or
		\item $l(x) = d_1 1 0^{k} 1 a_1 d_2$ for some $a' d_1 =  a_1 \in A$.
	\end{enumerate}
	\emph{Claim}: For $x, y \sqsubseteq g$, if $x \sim y$, then $x = y$. 
	Otherwise, there are $x, y \sqsubseteq g$, such that $l(x) = l(y)$, but $x \neq y$.
	We have the following cases:
	\begin{description}
		\item[$|x|, |y| \leq n$:] in this case, there are $n_1, n_2 \leq n$, such that $0^{n_1}x = 0^{n_2}y$; because $x$ and $y$ start with $1$, then $n_1 = n_2$ and $x=y$.
		\item[$|x| < n < |y|$ :] $y$ must be in one of the forms described above, so if $l(y) = 0^{n_1}1 a_1 1 a_210^{n_2}$, then $x = 1 a_1 1 a_210^{n_2}$, which is a contradiction, because $x$ is not in an appropriate form (right after $1 a_11$, there should be $0^k \neq a_2$); \\
		if $l(x) = l(y) = d_1 1 a_1 1 0^{k} 1 d_2$, then $10^k 1$ must appear exactly once in $x$, so $d_1 = 0^{n_1}$ and $a_1 = b_1$, meaning that $y$ is an initial fragment of $g$, thus $d_1 = \epsilon$ a contradiction, because $l(x)$ starts with $0$; \\
		if $l(x) = d_1 1 0^{k} 1 a_1 d_2$, then we already have a contradiction, because there is some $|a'|>k+1$, such that x starts with $a' 1 0^k 1$, so $d_1 = a'$, but $|d_1| \leq k$.
		\item[$|x|, |y| \geq n$: ]  $l(x) =l(y)$, so they must be of the same form; 
		if $l(x) = l(y) = 0^{n_1}1 a_1 1 a_210^{n_2}$, then $10^k 1 a_1 1 a_210^{n_2}$ are the last $n+n_2$ positions of $x$ and of $y$, so if $x \neq y$, then we found two different places in $g$ where $10^k 1 a_1$ appears, a contradiction; the cases of the other forms are similar.
	\end{description}

	Now, if $g$ is not simple, then by Lemma \ref{lem:find_px_det}, there are $x \sqsubset y \sqsubseteq g$, such that $p \xRightarrow{x} q$ and $p \xRightarrow{y} q$. Furthermore, by Corollary \ref{cor:det_mon_is_det}, if $p \xRightarrow{x} q_1$ and $p \xRightarrow{y} q_2$, $q_1 \sim q \sim q_2$. 
	If $x z \in M_n$, then $p \xRightarrow{x} r \xRightarrow{z} \yes$, so $p \xRightarrow{y} r'$, where $r' \sim r$, so $r' \xRightarrow{z} \yes$, meaning that $y z \in M_n$.
	Similarly, if $y z \in M_n$, then $x z \in M_n$, therefore $x \sim y$.
	Finally, since $x \neq y$ and $x \sim y$, we have a contradiction by the claim we proved above, so $g$ is simple and the proof complete.
\qed\end{proof}

Language $M_n$ above demonstrates an exponential gap between the state size of NFAs, (DFAs,) nondeterministic monitors, and deterministic monitors. 
Therefore, the upper bounds provided by Theorem \ref{thm:nfa_to_monitor} and Corollary \ref{cor:dfa_to_monitor} cannot be  improved significantly. 
On the other hand, it is not clear what the gap between a nondeterministic monitor and an equivalent deterministic one has to be. Corollary \ref{cor:determinization_upper} informs us that for every monitor of size $n$, there is an equivalent deterministic one of size $2^{O(2^n)}$; on the other hand, Theorem \ref{thm:MnisbadforDmon} presents a language (namely $M_n$) which is recognized by a (nondeterministic) monitor of size $k$ ($=O(2^n)$), but by no deterministic monitor of size $2^{o(k)}$. These bounds are significantly different, as there is an exponential gap between them, and they raise the question whether there is a more sophisticated (and efficient) procedure than turning a monitor into an NFA, then using the usual subset construction, and then turning the resulting DFA back into a monitor, as was done in Corollary \ref{cor:determinization_upper}. Theorem \ref{thm:determinizing_monitors_is_VERY_hard} demonstrates that the upper bound of Corollary \ref{cor:determinization_upper} cannot be improved much.

\begin{theorem}\label{thm:determinizing_monitors_is_VERY_hard}
	For every $n \in \NN$, there is an irrevocable regular language on two symbols which is recognized by a nondeterministic monitor of size $O(n)$ and by no deterministic monitor of size $2^{2^{o\left(\sqrt{n \log n}\right)}}$.
\end{theorem}

\begin{proof}
	For a word $t \in \{0,1\}^*$, we define the projections of $t$, $t|_0$ and $t|_1$ to be the result of removing all 1's, respectively all 0's, from $t$: for $i \in \{0,1\}$, $\epsilon|_i = \epsilon$, $i t' |_i = i (t'|_i)$, and $(1-i)t'|_i = t'|_i$.
	For $n \in \NN$, let $$F(n) = \max_{m_1 + \cdots + m_k = n} lcm(m_1,\ldots,m_k),$$
	where $lcm(m_1,\ldots,m_k)$ is the least common multiple of  $m_1,\ldots,m_k$,
	and $X(n) = \{m_1,\ldots,m_k\}$, where $m_1 + \cdots + m_k = n$ and  $lcm(m_1,\ldots,m_k) = F(n)$.

	The proof of Theorem \ref{thm:determinizing_monitors_is_VERY_hard} is based on a result by Chrobak
	\cite{CHROBAK1986149} (errata at \cite{CHROBAK2003497}), who demonstrated that for every natural number $n$, there is a unary language (a language over exactly one symbol) which is recognized by an NFA with $n$ states, but by no DFA with $e^{o\left(\sqrt{n \log n}\right)}$ states. The unique symbol used can be (in our case) either $0$ or $1$. We can use $1$, unless we make explicit otherwise.
	The unary language Chrobak introduced was $$Ch_n = \{1^{cm} \mid m \in X(n),\ 0 < c \in \NN \}.$$

	For some $n \in \NN$, let $Ch_n^0$ be Chrobak's language, where the symbol used is $0$ and $Ch_n^1$ be the same language, but with $1$ being the symbol used --- so for $i \in \{0,1\}$, 
	$Ch_n^i = \{i^{cm} \mid m \in X(n),\ 0 < c \in \NN \}$. Now, let 
	$$
	U_n = \{x e y \in \{0,1,e\}^* \mid x \in \{0,1\}^* \text{ and for some } i \in \{0,1\},\ x|_i \in Ch_n^i \}.
	$$
	Fix some $n \in \NN$ and let $X(n) = \{ m_1, m_2, \ldots , m_k \}$.
	$U_n$ can be recognized by the monitor $p_0 + p_1$ of size $O(n)$, where for $\{i,\overline{i}\} = \{0,1\}$, $p_i = p_i^1 + p_i^2 + \cdots + p_i^k$, where for $1 \leq j \leq k$, $p_i^j$ is the monitor defined recursively in the following way. Let 
	$$p_i^j[m_j] = \rec\ x_{m_j}.(\overline{i}.x_{m_l} + i.x_1 + e.\yes);$$ 
	for $0 \leq l < m_j$, let 
	$$p_i^j[l] = \rec\ x_l.(\overline{i}.x_l + i.p_i^j[l+1]);$$ 
	finally, let $p_i^j = p_i^j[0]$. That is, after putting everything together and simplifying the variable indexes,
	$$
	p_i^j = 
	\rec\ x_0.(\overline{i}. x_0 + \underbrace{i.\rec\ x_1. (\overline{i}. x_1 +  
		\cdots i.\rec\ x_{m_j}. (\overline{i}. x_{m_j}}_{m_j}  + i.x_1 + e.\yes) \cdots ))).\footnote{
		Notice that variable $x$ appears in several scopes of $\rec\ x$ and thus violates our agreement that each variable is bound by a unique recursion operator (and the same can be said for $x_b$). However, this can easily be dealt with by renaming the variables and it helps ease notation.}
	$$
	Monitor $p_i^j$ essentially ignores all appearances of $\overline{i}$ and counts how many times $i$ has appeared. If $i$ has appeared a multiple of $m_j$ times, then $p_i^j$ is given a chance to reach verdict $\yes$ if $e$ then appears; otherwise it continues counting from the beginning. That the size of $p_0 + p_1$ is $O(n)$ is evident from the fact that for every $i,j$, $|p_i^j| = O(m_j)$, which is not hard to calculate since $\left| p_i^j[m_j]\right| = 9$ and for $l < m_j$, $\left|p_i^j[l]\right| = \left|p_i^j[l+1]\right| + 5$.
	
	Let $p$ be a deterministic monitor for $U_n$ and $t \in \{0,1\}^*$. 
	To complete the proof of the theorem,  it suffices to prove for some constant $c > 0$ that if $|t|< e^{c  \cdot\sqrt{n \log n}}$, then $t$ must be simple. 
	Indeed, proving the above would mean that we have constructed a simple set of traces of cardinality more than $2^{ e^{c  \cdot\sqrt{n \log n}}}$ --- that is, the set of traces on $\{0,1\}$ of length less than $ e^{c  \cdot\sqrt{n \log n}}$ --- which, by Lemma \ref{lem:simple_traces}, gives the same lower bound for $|p|$.
	Specifically, we prove that if $t$ is not simple for $p$, then there is a DFA for Chrobak's unary language, $Ch_n$, of at most $|t|$ states, so by Chrobak's results it cannot be that $t$ is not simple and  $|t|<  e^{c  \cdot \sqrt{n \log n}}$. 
	
	So, let $t$ be a shortest non-simple trace on $\{0,1\}$. 
	Since $t$ is not simple and is minimal, $t = t_1 t_2 i$, where $i = 0$ or $1$ and $p \xRightarrow{t_1} r \xRightarrow{t_2} s \xRightarrow{i} r$ (by Lemma \ref{lem:find_px_det}). Without loss of generality, we assume $i = 1$. Let $A$ be the DFA $(Q,\{1 \},\delta,p,F)$, where $Q$ is the set of all monitors appearing in the derivation $p \xRightarrow{t_1} r \xRightarrow{t_2} s \xRightarrow{1} r$, $\delta(q_1,1) = q_2$ iff $q_1 \left(\xRightarrow{0}\right)^*\left(\xrightarrow{\tau}\right)^*\xrightarrow{1}q_2$ is part of the derivation above, and 
	$$F = \{q \in Q \mid \text{for all } c \geq 0,\ q \xRightarrow{0^c e} \yes \}.$$
	Notice that $A$ is, indeed, deterministic, since transitions move along the derivation and all monitors in $Q$ transition exactly once in $p \xRightarrow{t_1} r \xRightarrow{t_2} s \xRightarrow{1} r$ --- so the first $\xrightarrow{1}$ transition that appears after a state/monitor in the derivation exists and is unique.
	
	We claim that $A$ accepts $Ch_n$.
	By the definition of $\delta$, if the run of $A$ on $w$ reaches state (submonitor) $q$, then 
	there is some $w^t \in \{0,1\}^*$, such that $w^t|_1 = w$ and $p \xRightarrow{w^t} q$. Then, \\
	$A$ accepts $w$  \hfill  iff  \hfill 
	$q \in F$ 
	\hfill 	iff  \hfill  
	for all $c \geq 0$,  
	$q \xRightarrow{0^c e} \yes$ 
	\hfill 	iff  \hfill  \\
	for all $c \geq 0$,  
	$p \xRightarrow{w^t 0^c e} \yes$ (by Corollary \ref{cor:determinism_gives_always_same_v})
	\hfill  iff  \hfill 
	for every $c \geq 0$, $w^t 0^c e \in U_n$
	\hfill  iff  \hfill 
	$w \in Ch_n$.
	
	Finally, although we promised that we would only use two symbols, we have used three. 
	However, a language of three symbols can easily be encoded as one of two symbols, using the mapping: $0 \mapsto 00$, $1 \mapsto 01$, and $e \mapsto 11$. We can encode $U_n$ like this and simply continue with the remaining of the proof.
\qed\end{proof}

Notice that the lower bound given by Theorem \ref{thm:determinizing_monitors_is_VERY_hard} depends on the assumption that we can use two symbols in our regular language; this can be observed both from the statement of the theorem and from its proof, which makes non-trivial use of the two symbols. 
So, a natural question to ask is whether the same bounds hold for NFAs and monitors on one  symbol.

Consider an irrevocable regular language on one symbol. If $k$ is the length of the shortest word in the language, then we can immediately observe two facts. The first is that the smallest NFA which recognizes the language must have at least $k + 1$ states (and indeed, $k + 1$ states are enough). The second fact we can observe is that there is a deterministic monitor of size exactly $k+1$ which recognizes the language: $1^k.\yes$. Therefore, things are significantly easier when working with unary languages.

\begin{corollary}\label{cor:unary_is_easy}
	If there is an irrevocable NFA of $n$ states which recognizes unary language $L$,
	then, there is a deterministic monitor of size at most $n$ which recognizes $L$.
\end{corollary}

\section{Determinizing with Two Verdicts}

\label{sec:quick_fix}

In Section \ref{sec:bounds} we have dealt with monitors which can only reach a positive verdict or a negative one but not both. This was mainly done for convenience, since a single-verdict monitor is a lot easier to associate with a finite automaton and it helped simplify several cases. It is also worth mentioning that, as demonstrated in \cite{rvpaper}, to monitor for $\mhml$-properties, we are interested in single-verdict monitors. In this section, we demonstrate how the constructions and bounds of Section \ref{sec:bounds} transfer to the general case of monitors.

First, notice that there is no deterministic monitor equivalent to the monitor 
$m_c = \alpha.\yes + \alpha.\no$, also defined in Subsection \ref{subsec:det_verd_choices}, since $m_c$ can reach both verdicts with the same trace. Thus, there are monitors, which are not, in fact,  equivalent to deterministic ones. These are the ones for which there is a trace through which they can transition to both verdicts. We call these monitors \emph{conflicting}.

To treat non-conflicting monitors, we reduce the problem to the determinization of single-verdict monitors. For this, we define two transformations, very similar to what we did in Section \ref{sec:rewriting} to reduce the determinization of monitors to the determinization of processes. Let $[no]$ be a new action, not in \mycap{Act}. We define $\nu$  in the following way:
	\begin{align*}
&\nu(\no)                   = [no].\yes               , \\
&\nu(x)                      = x             , \\
&\nu(\alpha.m)               = \alpha.\nu(m)                , \\
&\nu(m + n)                  = \nu(m) + \nu(n)              ,\text{ and}  \\
&\nu(\texttt{rec} x.m)       = \texttt{rec} x.\nu(m)        .\\ \\
	\noalign{We also define  $\nu^{-1}$: if $s$ is a sum of $[no].\yes$, then $\nu^{-1}(s) = \no$ and otherwise, }
&\nu^{-1}(x)                      = x             , \\
&\nu^{-1}(\alpha.m)               = \alpha.\nu^{-1}(m)                ,\\ 
&\nu^{-1}(m + n)                  = \nu^{-1}(m) + \nu^{-1}(n)              , \text{ and}\\
&\nu^{-1}(\texttt{rec} x.m)       = \texttt{rec} x.\nu^{-1}(m)        .
\end{align*}

\begin{lemma}\label{lem:substitute_no}
	Let $q = \nu(p)$.
	Then, $p \xRightarrow{t} \no$ if and only if there is some $t' \sqsubseteq t$, such that $q \xRightarrow{t'} \no.r$.
\end{lemma}

\begin{proof}
	Straightforward induction on the derivations.
\qed\end{proof}

\begin{lemma}\label{lem:substitute_no_nil}
	Let $q = \nu^{-1}(p)$, where verdict \no\ does not appear in $p$.
	Then, there is some $t' \sqsubseteq t$ and a sum $s$ of $[no].\yes$, such that $p \xRightarrow{t'} s$, if and only if $q \xRightarrow{t} \no$.
\end{lemma}

\begin{proof}
	Straightforward induction on the derivations.
\qed\end{proof}

\begin{theorem}\label{thm:multi_verdict}
	Let $m$ be a monitor which is not conflicting. Then, there is an equivalent deterministic one of size $2^{2^{O(|m|)}}$.
\end{theorem}

\begin{proof}
	Let $m$ be a monitor which is not conflicting, but uses both verdicts \yes\ and \no. 
By Lemma \ref{lem:no_sums_of_verdicts}, we can assume that there are no sums of \no\ in $m$.
Let $m'$ be the result of replacing \no\ in $m$ by $[no].\yes$, where $[no]$ is a new action not appearing in $m$. Then, for some constant $c > 0$, $m'$ is a single-verdict monitor of size $c \cdot |m|$, so as we have shown in the preceding sections (Corollary \ref{cor:determinization_upper}), $m'$ is equivalent to a deterministic monitor $n'$ of size $2^{O(2^{c\cdot |m|})}$. Let $n$ be the result of replacing all maximal sums of $[no].\yes$ by $\no$. Then, $n$ is deterministic, because there are no sums of $\no$ (they would have been replaced) and all other sums remain in the form required by deterministic monitors. It remains to demonstrate that $m \meq n$, which we now do. Let $t \in \mycap{Act}^*$. Then,  \\
$m \xRightarrow{t} \no$ \\
iff 
$m \xRightarrow{t} \no$ and $m \not\xRightarrow{t} \yes$
($m$ is not conflicting) \\
iff
there is a $t' \sqsubseteq t$ and $\alpha \in \mycap{Act}$, such that 
$m' \xRightarrow{t'} [no] . \yes$ and  $m' \not\xRightarrow{t'\alpha} \yes$ (by Lemma \ref{lem:substitute_no})
\\
iff 
there is a $t' \sqsubseteq t$ and   $\alpha \in \mycap{Act}$, such that
$m' \xRightarrow{t'[no]} \yes$ and $m' \not\xRightarrow{t'\alpha} \yes$ 
(there are no sums of \no\ in $m$)
\\
iff
there is a $t' \sqsubseteq t$ and $\alpha \in \mycap{Act}$, such that
$n' \xRightarrow{t'[no]} \yes$ and $n' \not\xRightarrow{t'\alpha} \yes$ 
($m'$ and $n'$ are equivalent)
\\
iff
there is a $t' \sqsubseteq t$, an $\alpha \in \mycap{Act}$, and some $s \neq \yes$, such that
$n' \xRightarrow{t'} s \xRightarrow{[no]} \yes$ and $n' \not\xRightarrow{t'\alpha} \yes$ 
\\
iff
there is a $t' \sqsubseteq t$, an $\alpha \in \mycap{Act}$, and a sum $s$ of $[no] . \yes$, such that
$n' \xRightarrow{t'} s$ and $n' \not\xRightarrow{t'\alpha} \yes$ 
\\
iff 
there is a $t' \sqsubseteq t$ and $\alpha \in \mycap{Act}$, such that
$n \xRightarrow{t'} \no$ and $n \not\xRightarrow{t'\alpha} \yes$
(by Lemma \ref{lem:substitute_no_nil})
\\
iff 
$n \xRightarrow{t} \no$.\\
The case of the  $\yes$ verdict is straightforward.
\qed\end{proof}

Naturally, the same lower bounds as for single-verdict monitors  hold for the general case of monitors as well.

\subsubsection*{Conflicting Monitors}

We demonstrated in this section how we can determinize any non-conflicting monitor. A deterministic and conflicting monitor is a contradiction, as it would have to deterministically reach two different verdicts on the same trace. It would be good, therefore, to be able to detect conflicting monitors. Here we sketch how this can be done 
using nondeterministic logarithmic space.

For any monitor $m$, let $G_m = (V,E)$ be a graph, such that 
$$V = \{ (p,q) \mid p,q \text{ submonitors of } m \} \ \ \text{ and }$$
$$E = \{ (p,q,p',q') \in V^2 \mid \exists \alpha \in \mycap{Act}.\ p \xRightarrow{\alpha} p' \text{ and } q \xRightarrow{\alpha} q' \}.$$
Then, $m$ is conflicting iff there is a path from $(m,m)$ to $(\yes,\no)$ in $G_m$. This is a subproblem of the $st$-connectivity problem, known to be in \NL\ and solvable in  time $O(|V|+|E|) = O(|m|^4)$ (by running a search algorithm in $G_m$).

\section{Conclusions}
\label{sec:conclusions}

We have provided three methods for determinizing monitors. One of them is by reducing the problem to the determinization of processes, which has been handled by Rabinovich in \cite{rabinovich}; another is by using Rabinovich's methods directly on the formulae of $\muhml$, bringing them to a deterministic form, and then employing Francalanza et al.'s  monitor synthesis method from \cite{rvpaper}, ending up with a deterministic monitor; the last one transforms a monitor into an NFA, uses the classical subset construction from Finite Automata Theory to determinize the NFA, and then, from the resulting DFA constructs a deterministic monitor.

The first method is probably the simplest and directly gives results, at the same time describing how the behavior of monitors is linked to processes. The second method is more explicit, in that it directly applies Rabinovich's methods; furthermore, it is used directly on a $\muhml$ formula and helps us relate the (non-)deterministic  behavior of monitors to the form of the formula they were constructed to monitor.
The third method links monitors to finite automata and allows us to extract more precise bounds on the size of the constructed deterministic monitors.

Monitors are central for the runtime verification of processes. We have focused on monitors for $\muhml$, as constructed in \cite{rvpaper}.
We showed that we can add a runtime monitor to a system without having a significant effect on the execution time of the system. 
Indeed, in general, evaluating a nondeterministic monitor of size $n$ for some specific trace 
amounts to keeping track of all possible derivations along the trace.
This can take exponential time with regards to the size of the monitor, as $n$ submonitors give rise to  an exponential number of combinations of submonitors we could reach along a trace, if the trace is long enough. 
More importantly, computing each transition can take time up to $n^2$, because for up to $n$ submonitors we may need to transition to up to $n$ submonitors.
By using a deterministic monitor, each transition is provided explicitly by the monitor description, so we can transition immediately at every step along a trace --- with a cost depending on the implementation.
This speed-up can come at a severe cost, though, since we may have to use up to doubly-exponential more space to store the monitor, and even stored in a more efficient form as its LTS, the deterministic monitor may require exponential extra space.

\paragraph{Summary of Bounds:}

We were able to prove certain upper and lower bounds for several constructions. Here we give a summary of the bounds we have proven, the bounds which are known, and the ones we can further infer from these results.
\begin{itemize}
	\item 
	Corollary \ref{cor:determinization_upper} (actually, Theorem \ref{thm:multi_verdict} for the general case) informs us that from a nondeterministic monitor of size $n$, we can construct a deterministic one of size $2^{O(2^n)}$.
	\item 
	Theorem \ref{thm:determinizing_monitors_is_VERY_hard} explains that we cannot do much better, because there is an infinite family of monitors, such that for any monitor of size $n$ in the family, there is no equivalent deterministic monitor of size $2^{2^{o(\sqrt{n \log n})}}$.
	\item 
	As for when we start with an NFA, it is a classical result
	that an NFA of $n$ states is equivalent to a DFA of $2^n$ states; furthermore, this bound is tight.
	\item
	Theorem \ref{thm:nfa_to_monitor} informs us that an irrevocable NFA of $n$ states can be converted to an equivalent monitor of size $2^{O(n \log n)}$.
	\item 
	Proposition \ref{prp:nfa2mon_is_hard} reveals that there is an infinite family of NFAs, for which every NFA of the family of $n$ states is not equivalent to any monitor of size $2^{o(n)}$.
	\item
	Corollary \ref{cor:determinization_upper} yields that an irrevocable NFA of $n$ states can be converted to an equivalent deterministic monitor of size $2^{O(2^n)}$; Theorem \ref{thm:MnisbadforDmon} makes this bound tight.
	\item 
	Corollary \ref{cor:dfa_to_monitor} allows us to convert a DFA of $n$ states to a deterministic monitor of $2^{O(n)}$ states; Theorem \ref{thm:MnisbadforDmon} makes this bound tight.
	\item 
	This $2^{O(n)}$ is also the best upper bound we have for converting a DFA to a (general) monitor; it is unclear at this point what lower bounds we can establish for this transformation.
	\item 
	We can convert a (single-verdict) monitor of size $n$ to an equivalent DFA of $O(2^n)$ states, by first converting the monitor to an NFA of $n$ states (Proposition \ref{prp:A(p)is_good}) and then using the classical subset construction.
	\item 
	If we could convert any monitor of size $n$ to a DFA of $2^{o(\sqrt{n \log n})}$ states, then we could use the construction of Corollary \ref{cor:dfa_to_monitor} to construct a deterministic monitor of $2^{2^{o(\sqrt{n \log n})}}$ states, which contradicts the lower bound of Theorem \ref{thm:determinizing_monitors_is_VERY_hard}; therefore, $2^{\Omega(\sqrt{n \log n})}$ is a lower bound for converting monitors to equivalent DFAs.
	\item
	Using Lemma \ref{lemma:trace_equiv}, we can reduce the problem of determinizing monitors to the determinization of processes (up to trace-equivalence); by Theorem \ref{thm:determinizing_monitors_is_VERY_hard}, this gives a lower bound of $2^{2^{\Omega(\sqrt{n \log n})}}$ for Rabinovich's construction in \cite{rabinovich}.
	\item
	Similarly, using the constructions of \cite{rvpaper}, one can convert a $\mhml$ formula into a monitor for it and a monitor into a $\mhml$ formula of the same size (or smaller). Therefore, we can conclude that the lower bounds for determinizing monitors also hold for determinizing $\mhml$ formulas as in Subsection \ref{sect:form_rewrite}.
	Therefore, a $\chml$ formula which holds precisely for the traces in language $M_n$ from Section \ref{sec:bounds} must be of size $2^{\Omega(n)}$ and an equivalent deterministic \chml\ formula must be of size $2^{2^{\Omega(n)}}$.
	Therefore, NFAs 
	as a specification language, can be exponentially more succinct than the monitorable fragment of $\muhml$ and doubly exponentially more succinct than the deterministic monitorable fragment of $\muhml$; DFAs can be exponentially more succinct than the deterministic monitorable fragment of $\muhml$.
	\item 
	Corollary \ref{cor:unary_is_easy} informs us that it is significantly easier to convert an irrevocable NFA or monitor into a (deterministic) monitor when the alphabet we  use (the set of actions) is limited to one element: when there is only one action/symbol,  an irrevocable NFA of $n$ states or nondeterministic monitor of size $n$ can be converted into an equivalent deterministic monitor of size at most $n$.
	\item
	In Section \ref{sec:quick_fix}, we have argued that detecting whether a monitor is conflicting can be done in time $O(n^4)$ or in nondeterministic space $O(\log n)$ (and thus, by Savitch's Theorem \cite{SAVITCH1970177} in deterministic space $O(\log^2 n)$).
\end{itemize}
The bounds we were able to prove can be found in Table \ref{tab:bounds}.

\begin{table}[h] \centering
	\begin{tabular}{|r|c|c|c|}
		\hline
		from/to		&	DFA	&	monitor	&	det. monitor \\ \hline
		NFA			&	tight: $O(2^n)$	&	\begin{minipage}{.25\textwidth} \vspace{1ex}
											\begin{tabular}{c} 
											upper: $2^{O(n \log n)}$	\\	
											lower: $2^{\Omega(n)}$
											\end{tabular}
											\end{minipage}&	
														tight:	$2^{O(2^n)}$	\\ \hline
		DFA			&	X	&		upper: $2^{O(n)}$	& tight:	$2^{O(n)}$		\\ \hline
		\begin{tabular}{r}
		nondet. \\ monitor
		\end{tabular}			
					&	
					\begin{minipage}{.25\textwidth} \vspace{1ex}
					\begin{tabular}{c} 
						upper: $O(2^{n})$	\\	
						lower: $2^{\Omega(\sqrt{n \log n})}$
					\end{tabular}	
					\end{minipage}
							&	X		
										&	
											\begin{minipage}{.25\textwidth} \vspace{1ex}
											\begin{tabular}{c}
											upper: $2^{O(2^n)}$ \\
											lower: $2^{2^{\Omega(\sqrt{n \log n})}}$	\\
											\end{tabular}
											\end{minipage}\\
		\hline
	\end{tabular}
	\caption{Bounds on the cost of construction (X signifies that the conversion is trivial)}
	\label{tab:bounds}
\end{table}

\paragraph{Optimizations:}

Monitors to be used for the runtime verification of processes are expected to not affect the systems they monitor as much as possible. Therefore, the efficiency of monitoring must be taken into account to restrict overhead.
To use a deterministic monitor for a $\muhml$ property, we would naturally want to keep its size as small as possible. It would help to preserve space (and time for each transition) to store the monitor in its LTS form --- as a DFA. 
We should also aim to use the smallest possible monitor we can. There are efficient methods for minimizing a DFA, so one can use these to find a minimal DFA and then turn it into monitor form using the construction from Theorem \ref{thm:nfa_to_monitor}, if such a form is required. The resulting monitor will be (asymptotically) minimal:

\begin{proposition}
	Let $A$ be a minimal DFA for an irrevocable language $L$, such that $A$ has $n$ states and there are at least $\Pi$ paths in $A$ originating at its initial state. Then, there is no deterministic monitor of size less than $\Pi$, which recognizes $L$.
\end{proposition}

\begin{proof}
	Since $A$ is deterministic, all paths in $A$ are completely described by a trace $t \in \mycap{Act}^*$. We show that for every deterministic monitor $p$ which recognizes $L$, such a  $t$ is simple. Let $t$ be the shortest trace which gives a path in $A$, but is not simple for $p$. By Lemma \ref{lem:find_px_det}, there are $t' u = t$, such that $|u|>0$ and $p \xRightarrow{t'} r \xRightarrow{u} r$. Since $t$ represents a path in $A$ and $A$ is deterministic, $A$ can reach state $q$ with trace $t$ and $q' \neq q$ with trace $t'$. Since $A$ is a minimal DFA, there must be a trace $s$, such that (without loss of generality) $A$  reaches an accepting state from $q$ through $s$ and a non-accepting state from $q'$ through $s$. Therefore $ts \in L$ and $t' s \notin L$, which is a contradiction, because by Corollary \ref{cor:determinism_gives_always_same_v} and the observation above, $p \xRightarrow{ts} \yes$ iff $p \xRightarrow{t's} \yes$.
\qed\end{proof}

As we see, DFA minimization also solves the problem of deterministic monitor minimization. On the other hand, it would be good to keep things small from an earlier point of the construction, before the exponential explosion of states of the subset construction takes place. In other words, it would be good to minimize the NFA we construct from the monitor, which can already be smaller than the original monitor. Unfortunately, NFA minimization is a hard problem --- specifically \PSPACE-complete \cite{zbMATH00495784} --- and it remains \NP-hard even for classes of NFAs which are very close to DFAs \cite{bjorklund2012tractability}. NFA minimization is even hard to approximate or parameterize \cite{GRAMLICH2007908,GruberH07}. Still, it would be better to use an efficient approximation algorithm from \cite{GruberH07} to process the NFA and save on the number of states before we determinize.
This raises the question of whether (nondeterministic) monitors are easier to minimize than NFAs, although a positive answer seems unlikely in light of the hardness results for NFA minimization.

\bibliographystyle{plain}   
\bibliography{references}

\begin{thebibliography}{10}

\bibitem{AceILS:2007}
Luca Aceto, Anna Ing\'{o}lfsd\'{o}ttir, Kim~Guldstrand Larsen, and Jiri Srba.
\newblock {\em Reactive Systems: Modelling, Specification and Verification}.
\newblock Cambridge Univ. Press, New York, NY, USA, 2007.

\bibitem{arnold2001rudiments}
A.~Arnold and D.~Niwinski.
\newblock {\em Rudiments of {$\mu$}-Calculus}.
\newblock Studies in Logic and the Foundations of Mathematics. Elsevier
  Science, 2001.

\bibitem{goodbadugly}
Andreas Bauer, Martin Leucker, and Christian Schallhart.
\newblock The good, the bad, and the ugly, but how ugly is ugly?
\newblock In {\em RV}, volume 4839 of {\em LNCS}, pages 126--138. Springer,
  2007.

\bibitem{comparingLTL}
Andreas Bauer, Martin Leucker, and Christian Schallhart.
\newblock Comparing {LTL} semantics for runtime verification.
\newblock {\em Logic and Comput.}, 20(3):651--674, 2010.

\bibitem{rvLTL}
Andreas Bauer, Martin Leucker, and Christian Schallhart.
\newblock Runtime verification for {LTL} and {TLTL}.
\newblock {\em TOSEM}, 20(4):14, 2011.

\bibitem{bjorklund2012tractability}
Henrik Bj{\"o}rklund and Wim Martens.
\newblock The tractability frontier for {NFA} minimization.
\newblock {\em Journal of Computer and System Sciences}, 78(1):198--210, 2012.

\bibitem{CassarF16}
Ian Cassar and Adrian Francalanza.
\newblock On implementing a monitor-oriented programming framework for actor
  systems.
\newblock In Erika {\'{A}}brah{\'{a}}m and Marieke Huisman, editors, {\em
  Integrated Formal Methods - 12th International Conference, {iFM} 2016},
  volume 9681 of {\em Lecture Notes in Computer Science}, pages 176--192.
  Springer, 2016.

\bibitem{CHROBAK1986149}
Marek Chrobak.
\newblock Finite automata and unary languages.
\newblock {\em Theoretical Computer Science}, 47:149--158, 1986.

\bibitem{CHROBAK2003497}
Marek Chrobak.
\newblock Errata to: ``{Finite Automata and Unary Languages}''.
\newblock {\em Theoretical Computer Science}, 302(1):497--498, 2003.

\bibitem{clarke1981design}
Edmund~M Clarke and E~Allen Emerson.
\newblock Design and synthesis of synchronization skeletons using branching
  time temporal logic.
\newblock In {\em Workshop on Logic of Programs}, pages 52--71. Springer, 1981.

\bibitem{dAmorimR05}
Marcelo d'Amorim and Grigore Rosu.
\newblock Efficient monitoring of $\omega$-languages.
\newblock In Kousha Etessami and Sriram~K. Rajamani, editors, {\em Computer
  Aided Verification, 17th International Conference, {CAV} 2005}, volume 3576
  of {\em Lecture Notes in Computer Science}, pages 364--378. Springer, 2005.

\bibitem{DeboisHS15}
S{\o}ren Debois, Thomas~T. Hildebrandt, and Tijs Slaats.
\newblock Safety, liveness and run-time refinement for modular process-aware
  information systems with dynamic sub processes.
\newblock In Nikolaj Bj{\o}rner and Frank~S. de~Boer, editors, {\em {FM} 2015:
  Formal Methods - 20th International Symposium}, volume 9109 of {\em Lecture
  Notes in Computer Science}, pages 143--160. Springer, 2015.

\bibitem{truncated}
Cindy Eisner, Dana Fisman, John Havlicek, Yoad Lustig, Anthony McIsaac, and
  David~Van Campenhout.
\newblock Reasoning with temporal logic on truncated paths.
\newblock In {\em CAV}, volume 2725 of {\em LNCS}, pages 27--39. Springer,
  2003.

\bibitem{Erling04:inlined}
Ulfar Erlingsson.
\newblock {\em The {I}nlined {R}eference {M}onitor approach to {S}ecurity
  {P}olicy {E}nforcement}.
\newblock PhD thesis, Cornell University, 2004.

\bibitem{FalconeFM12}
Yli{\`e}s Falcone, Jean-Claude Fernandez, and Laurent Mounier.
\newblock What can you verify and enforce at runtime?
\newblock {\em STTT}, 14(3):349--382, 2012.

\bibitem{rvpaper}
Adrian Francalanza, Luca Aceto, and Anna Ingolfsdottir.
\newblock On verifying {Hennessy-Milner} logic with recursion at runtime.
\newblock In Ezio Bartocci and Rupak Majumdar, editors, {\em Runtime
  Verification}, volume 9333 of {\em Lecture Notes in Computer Science}, pages
  71--86. Springer International Publishing, 2015.

\bibitem{Geilen}
Marc Geilen.
\newblock {O}n the {C}onstruction of {M}onitors for {T}emporal {L}ogic
  {P}roperties.
\newblock In {\em RV}, volume~55 of {\em ENTCS}, pages 181--199, 2001.

\bibitem{GRAMLICH2007908}
Gregor Gramlich and Georg Schnitger.
\newblock Minimizing {NFA}'s and regular expressions.
\newblock {\em Journal of Computer and System Sciences}, 73(6):908--923, 2007.

\bibitem{Gray86}
Jim Gray.
\newblock Why do computers stop and what can be done about it?
\newblock In {\em Fifth Symposium on Reliability in Distributed Software and
  Database Systems, {SRDS} 1986}, pages 3--12. {IEEE} Computer Society, 1986.

\bibitem{GruberH07}
Hermann Gruber and Markus Holzer.
\newblock Inapproximability of nondeterministic state and transition complexity
  assuming {P$\neq$NP}.
\newblock In Tero Harju, Juhani Karhum{\"{a}}ki, and Arto Lepist{\"{o}},
  editors, {\em Developments in Language Theory, 11th International Conference,
  {DLT} 2007}, volume 4588 of {\em Lecture Notes in Computer Science}, pages
  205--216. Springer, 2007.

\bibitem{HeCL13}
Yuqin He, Xiangping Chen, and Ge~Lin.
\newblock Composition of monitoring components for on-demand construction of
  runtime model based on model synthesis.
\newblock In Hong Mei, Jian Lv, and Xiaoguang Mao, editors, {\em Proceedings of
  the 5th Asia-Pacific Symposium on Internetware, Internetware 2013}, pages
  20:1--20:4. {ACM}, 2013.

\bibitem{zbMATH00495784}
Tao {Jiang} and B.~{Ravikumar}.
\newblock {Minimal NFA problems are hard.}
\newblock {\em {SIAM J. Comput.}}, 22(6):1117--1141, 1993.

\bibitem{KleinG15}
John Klein and Ian Gorton.
\newblock Runtime performance challenges in big data systems.
\newblock In C.~Murray Woodside, editor, {\em Proceedings of the 2015 Workshop
  on Challenges in Performance Methods for Software Development, WOSP-C'15},
  pages 17--22. {ACM}, 2015.

\bibitem{KOZEN1983333}
Dexter Kozen.
\newblock Results on the propositional μ-calculus.
\newblock {\em Theoretical Computer Science}, 27(3):333--354, 1983.

\bibitem{Larsen90}
Kim~Guldstrand Larsen.
\newblock Proof systems for satisfiability in {Hennessy-Milner} logic with
  recursion.
\newblock {\em Theoretical Computer Science}, 72(2{\&}3):265--288, 1990.

\bibitem{Leu:RV:Overv}
Martin Leucker and Christian Schallhart.
\newblock A brief account of {R}untime {V}erification.
\newblock {\em JLAP}, 78(5):293 -- 303, 2009.

\bibitem{Ligatti05}
Jay Ligatti, Lujo Bauer, and David Walker.
\newblock Edit automata: enforcement mechanisms for run-time security policies.
\newblock {\em Int. J. Inf. Secur.}, 4(1-2):2--16, 2005.

\bibitem{LuoHEM14}
Qingzhou Luo, Farah Hariri, Lamyaa Eloussi, and Darko Marinov.
\newblock An empirical analysis of flaky tests.
\newblock In Shing{-}Chi Cheung, Alessandro Orso, and Margaret{-}Anne~D.
  Storey, editors, {\em Proceedings of the 22nd {ACM} {SIGSOFT} International
  Symposium on Foundations of Software Engineering, (FSE-22)}, pages 643--653.
  {ACM}, 2014.

\bibitem{MarinescuHC14}
Paul~Dan Marinescu, Petr Hosek, and Cristian Cadar.
\newblock Covrig: A framework for the analysis of code, test, and coverage
  evolution in real software.
\newblock In Corina~S. Pasareanu and Darko Marinov, editors, {\em International
  Symposium on Software Testing and Analysis, {ISSTA} '14}, pages 93--104.
  {ACM}, 2014.

\bibitem{MemonC04}
Atif~M. Memon and Myra~B. Cohen.
\newblock Automated testing of {GUI} applications: models, tools, and
  controlling flakiness.
\newblock In David Notkin, Betty H.~C. Cheng, and Klaus Pohl, editors, {\em
  35th International Conference on Software Engineering, {ICSE} '13}, pages
  1479--1480. {IEEE} Computer Society, 2013.

\bibitem{MJG+11mop}
Patrick~O'Neil Meredith, Dongyun Jin, Dennis Griffith, Feng Chen, and Grigore
  Ro\c{s}u.
\newblock An overview of the {MOP} runtime verification framework.
\newblock {\em STTT}, 14(3):249--289, 2012.

\bibitem{Milner:1982:CCS:539036}
R.~Milner.
\newblock {\em A Calculus of Communicating Systems}.
\newblock Springer-Verlag New York, Inc., Secaucus, NJ, USA, 1982.

\bibitem{pnueli1977temporal}
Amir Pnueli.
\newblock The temporal logic of programs.
\newblock In {\em Foundations of Computer Science, 1977., 18th Annual Symposium
  on}, pages 46--57. IEEE, 1977.

\bibitem{rabin1959finite}
Michael~O Rabin and Dana Scott.
\newblock Finite automata and their decision problems.
\newblock {\em IBM journal of research and development}, 3(2):114--125, 1959.

\bibitem{rabinovich}
Alexander Rabinovich.
\newblock A complete axiomatisation for trace congruence of finite state
  behaviors.
\newblock In {\em International Conference on Mathematical Foundations of
  Programming Semantics}, pages 530--543. Springer, 1993.

\bibitem{SAVITCH1970177}
Walter~J. Savitch.
\newblock Relationships between nondeterministic and deterministic tape
  complexities.
\newblock {\em Journal of Computer and System Sciences}, 4(2):177--192, 1970.

\bibitem{sipser1997introduction}
M.~Sipser.
\newblock {\em Introduction to the Theory of Computation}.
\newblock Computer Science Series. PWS Publishing Company, 1997.

\bibitem{VardiW94}
M.Y. Vardi and P.~Wolper.
\newblock Reasoning about infinite computations.
\newblock {\em Inf.\& Comp.}, 115(1):1--37, 1994.

\bibitem{Zhang:2014}
Sai Zhang, Darioush Jalali, Jochen Wuttke, K{\i}van\c{c} Mu\c{s}lu, Wing Lam,
  Michael~D. Ernst, and David Notkin.
\newblock Empirically revisiting the test independence assumption.
\newblock In {\em ISSTA}, 2014.

\end{thebibliography}

\end{document}